%% file: main.tex
\RequirePackage[l2tabu,orthodox]{nag}
\documentclass
[12pt,letterpaper]
{article}

\usepackage{etex}
\usepackage{xspace,enumerate}
\usepackage[dvipsnames]{xcolor}
\usepackage[T1]{fontenc}
\usepackage[full]{textcomp}
\usepackage[american]{babel}
\usepackage{mathtools}
\usepackage{amsthm}
\newtheorem{theorem}{Theorem}[section]
\newtheorem*{theorem*}{Theorem}

\newtheorem*{proposition*}{Proposition}
\newtheorem{lemma}[theorem]{Lemma}
\newtheorem*{lemma*}{Lemma}

\newtheorem*{conjecture*}{Conjecture}

\newtheorem*{fact*}{Fact}

\newtheorem*{hypothesis*}{Hypothesis}
\newtheorem{conjecture}[theorem]{Conjecture}
\theoremstyle{definition}
\newtheorem{definition}[theorem]{Definition}
\newtheorem*{definition*}{Definition}

\newtheorem{example}[theorem]{Example}
\newtheorem{question}[theorem]{Question}

\newtheorem{problem}[theorem]{Problem}

\theoremstyle{remark}
\newtheorem{claim}[theorem]{Claim}
\newtheorem*{claim*}{Claim}

\newtheorem*{remark*}{Remark}

\newtheorem*{observation*}{Observation}
\usepackage[
letterpaper,
top=0.8in,
bottom=0.8in,
left=1in,
right=1in]{geometry}
\usepackage{newpxtext} %
\usepackage{textcomp} %
\usepackage[varg,bigdelims]{newpxmath}
\usepackage[scr=rsfso]{mathalfa}%
\usepackage{bm} %
\let\mathbb\varmathbb
\usepackage{microtype}
\usepackage[
pagebackref,
colorlinks=true,
urlcolor=blue,
linkcolor=blue,
citecolor=OliveGreen,
]{hyperref}
\usepackage{amsthm}
\usepackage[capitalise,nameinlink]{cleveref}
\crefname{lemma}{Lemma}{Lemmas}
\crefname{question}{Question}{Questions}
\crefname{conjecture}{Conjecture}{Conjectures}
\crefname{fact}{Fact}{Facts}
\crefname{theorem}{Theorem}{Theorems}
\crefname{corollary}{Corollary}{Corollaries}
\crefname{claim}{Claim}{Claims}
\crefname{example}{Example}{Examples}
\crefname{algorithm}{Algorithm}{Algorithms}
\crefname{problem}{Problem}{Problems}
\crefname{definition}{Definition}{Definitions}
\usepackage{paralist}
\usepackage{turnstile}
\usepackage{mdframed}
\usepackage{tikz}

\newcommand{\Authornote}[2]{}
\newcommand{\Authornotecolored}[3]{}
\newcommand{\Authorcomment}[2]{}
\newcommand{\Authorfnote}[2]{}
\newcommand{\Pnote}{\Authornote{P}}

\newcommand{\Tnote}{\Authornotecolored{ForestGreen}{T}}

\newcommand{\Dnote}{\Authornote{D}}

\usepackage{boxedminipage}
\newcommand{\paren}[1]{(#1)}
\newcommand{\Paren}[1]{\left(#1\right)}

\newcommand{\abs}[1]{\lvert#1\rvert}

\newcommand{\card}[1]{\lvert#1\rvert}

\newcommand{\set}[1]{\{#1\}}
\newcommand{\Set}[1]{\left\{#1\right\}}

\newcommand{\norm}[1]{\lVert#1\rVert}
\newcommand{\Norm}[1]{\left\lVert#1\right\rVert}

\newcommand{\iprod}[1]{\langle#1\rangle}
\newcommand{\Iprod}[1]{\left\langle#1\right\rangle}

\newcommand{\Esymb}{\mathbb{E}}
\newcommand{\Psymb}{\mathbb{P}}

\DeclareMathOperator*{\E}{\Esymb}

\DeclareMathOperator*{\ProbOp}{\Psymb}
\renewcommand{\Pr}{\ProbOp}

\newcommand{\lmax}{\lambda_{\max}}

\newcommand{\bits}{\{0,1\}}
\newcommand{\sbits}{\{\pm1\}}
\newcommand{\defeq}{\stackrel{\mathrm{def}}=}
\newcommand{\seteq}{\mathrel{\mathop:}=}

\newcommand{\Mid}{\;\middle\vert\;}
\newcommand{\mper}{\,.}
\newcommand{\mcom}{\,,}
\newcommand\bdot\bullet
\DeclareMathOperator{\Ind}{\mathbf 1}

\DeclareMathOperator{\Tr}{Tr}

\DeclareMathOperator{\argmax}{argmax}

\DeclareMathOperator{\supp}{supp}

\DeclareMathOperator{\rank}{rank}

\newcommand{\etal}{et al.\xspace}

\newcommand{\Erdos}{Erd\H{o}s\xspace}
\newcommand{\Renyi}{R\'enyi\xspace}

\newcommand{\N}{\mathbb N}
\newcommand{\R}{\mathbb R}

\newcommand{\problemmacro}[1]{\texorpdfstring{\textup{\textsc{#1}}}{#1}\xspace}

\newcommand{\maxcut}{\problemmacro{max cut}}

\newcommand{\sparsestcut}{\problemmacro{sparsest cut}}

\newcommand{\densestksubgraph}{\problemmacro{densest $k$-subgraph}}
\newcommand{\cA}{\mathcal A}
\newcommand{\cB}{\mathcal B}

\newcommand{\cD}{\mathcal D}
\newcommand{\cE}{\mathcal E}

\newcommand{\cI}{\mathcal I}
\newcommand{\cJ}{\mathcal J}

\newcommand{\cM}{\mathcal M}
\newcommand{\cN}{\mathcal N}

\newcommand{\cP}{\mathcal P}

\newcommand{\cS}{\mathcal S}

\newcommand{\cX}{\mathcal X}
\newcommand{\cY}{\mathcal Y}

\renewcommand{\leq}{\leqslant}
\renewcommand{\le}{\leqslant}
\renewcommand{\geq}{\geqslant}
\renewcommand{\ge}{\geqslant}
\let\epsilon=\varepsilon
\numberwithin{equation}{section}
\newcommand\MYcurrentlabel{xxx}
\newcommand{\MYstore}[2]{%
  \global\expandafter \def \csname MYMEMORY #1 \endcsname{#2}%
}
\newcommand{\MYload}[1]{%
  \csname MYMEMORY #1 \endcsname%
}
\newcommand{\MYnewlabel}[1]{%
  \renewcommand\MYcurrentlabel{#1}%
  \MYoldlabel{#1}%
}
\newcommand{\MYdummylabel}[1]{}
\newcommand{\torestate}[1]{%
  \let\MYoldlabel\label%
  \let\label\MYnewlabel%
  #1%
  \MYstore{\MYcurrentlabel}{#1}%
  \let\label\MYoldlabel%
}
\newcommand{\restatetheorem}[1]{%
  \let\MYoldlabel\label
  \let\label\MYdummylabel
  \begin{theorem*}[Restatement of \cref{#1}]
    \MYload{#1}
  \end{theorem*}
  \let\label\MYoldlabel
}
\newcommand{\restatelemma}[1]{%
  \let\MYoldlabel\label
  \let\label\MYdummylabel
  \begin{lemma*}[Restatement of \cref{#1}]
    \MYload{#1}
  \end{lemma*}
  \let\label\MYoldlabel
}
\newcommand{\restateprop}[1]{%
  \let\MYoldlabel\label
  \let\label\MYdummylabel
  \begin{proposition*}[Restatement of \cref{#1}]
    \MYload{#1}
  \end{proposition*}
  \let\label\MYoldlabel
}
\newcommand{\restatefact}[1]{%
  \let\MYoldlabel\label
  \let\label\MYdummylabel
  \begin{fact*}[Restatement of \prettyref{#1}]
    \MYload{#1}
  \end{fact*}
  \let\label\MYoldlabel
}
\newcommand{\restate}[1]{%
  \let\MYoldlabel\label
  \let\label\MYdummylabel
  \MYload{#1}
  \let\label\MYoldlabel
}

\newcommand{\e}{\epsilon}
\newcommand{\eps}{\epsilon}

\allowdisplaybreaks
\newcommand*{\bT}{\mathbf{T}}
\newcommand*{\bY}{\mathbf{Y}}
\newcommand*{\bX}{\mathbf{X}}
\newcommand*{\xdist}{\sigma}
\newcommand*{\jdist}{\cJ}
\newcommand*{\jndist}{\cJ_{\emptyset}}
\newcommand*{\jfdist}{\cJ_{*}}

\newcommand*{\tensordecomposition}{tensor decomposition}
\newcommand*{\subdist}{\Upsilon}
\newcommand*{\planteddist}{\cD_{*}}

\newcommand*{\One}{1}
\newcommand*{\randomdist}{\cD_{\emptyset}}

\newcommand*{\gnhalf}{\mathbb{G}(n,\tfrac{1}{2})}
\newcommand*{\nulld}{\cD_{\emptyset}}
\newcommand*{\feasd}{\cD_{*}}
\newcommand*{\Id}{\mathrm{Id}}

\newcommand*{\normf}[1]{\Norm{#1}_{\mathrm{F}}}

\newcommand*{\kclique}{\problemmacro{$k$-clique}}

\newcommand*{\tensorpca}{\problemmacro{tensor PCA}}

\newcommand*{\fourxor}{\problemmacro{4-xor}}
\newcommand*{\F}{\mathbb{F}}

\newcommand*{\norminj}[1]{\norm{#1}_{\mathrm{inj}}}

\DeclareMathOperator*{\pE}{\tilde{\mathbb{E}}}
\DeclareMathOperator{\Span}{span}

\newcommand*{\transpose}[1]{{#1}{}^{\mkern-1.5mu\mathsf{T}}}
\newcommand*{\dyad}[1]{#1#1{}^{\mkern-1.5mu\mathsf{T}}}

\title{
  High-dimensional estimation\\ via sum-of-squares proofs
}

\author{
    Prasad Raghavendra\thanks{U.C. Berkeley. Supported by NSF CCF-1408643 and NSF CCF-1718695.}
      \and
    Tselil Schramm\thanks{MIT and Harvard. Supported by NSF awards CCF-1565264 and CNS-1618026, and the Simons Foundation.}
      \and
    David Steurer\thanks{ETH Z\"urich.}
}

\date{}

\begin{document}

\pagestyle{empty}

\maketitle
\thispagestyle{empty} %

\begin{abstract}

    \input{content/abstract}

\end{abstract}

\clearpage

  \microtypesetup{protrusion=false}
  \tableofcontents{}
  \microtypesetup{protrusion=true}

\clearpage

\pagestyle{plain}
\setcounter{page}{1}

    \input{content/introduction}
  
    \input{content/prelims}
  
    \input{content/algorithms}
    \input{content/lowerbounds}
  
    \input{content/spectral}
  
    \input{content/conclusion}

  \phantomsection
  \addcontentsline{toc}{section}{References}
  \bibliographystyle{amsalpha}
  \bibliography{bib/mathreview,bib/dblp,bib/custom,bib/scholar}

\appendix

    \input{content/algproofs}
  
    \input{content/lbproofs}

\end{document}

%% file: content/abstract.tex
Estimation is the computational task of recovering a \emph{hidden parameter} $x$ associated with a distribution $\cD_x$, given a \emph{measurement} $y$ sampled from the distribution.
High dimensional estimation problems arise naturally in statistics, machine learning, and complexity theory.

Many high dimensional estimation problems can be formulated as systems of polynomial equalities and inequalities, and thus give rise to natural probability distributions over polynomial systems.
Sum-of-squares proofs provide a powerful framework to reason about polynomial systems, and further there exist efficient algorithms to search for low-degree sum-of-squares proofs.

Understanding and characterizing the power of sum-of-squares proofs for estimation problems has been a subject of intense study in recent years.
On one hand, there is a growing body of work utilizing sum-of-squares proofs for recovering solutions to polynomial systems when the system is feasible.
On the other hand, a general technique referred to as \emph{pseudocalibration} has been developed towards showing lower bounds on the degree of sum-of-squares proofs.
Finally, the existence of sum-of-squares refutations of a polynomial system  has been shown to be intimately connected to the existence of spectral algorithms.
In this article we survey these developments.

%% file: content/introduction.tex
\section{Introduction}
\label{sec:introduction}

In estimation problems, the goal is to recover a structured object from an observed input which partially obfuscates it.
Formally, an estimation problem is specified by a family of distributions $\{\cD_x\}$ over $\R^N$ parametrized by $x \in \R^n$.
The input consists of a sample $y \in \R^N$ drawn from $\cD_x$ for some $x \in \R^n$, and the goal is to recover the value of the parameter $x$.
We refer to $x$ as the {\em hidden variable} or the {\em parameter}, and to  the sample $y$ as the {\em measurement} or the {\em instance}.

Often, it is information-theoretically impossible to recover hidden variables $x$ in that their value is not completely determined by the measurements.
Further, even when recovery is information-theoretically possible, in many high-dimensional settings it is computationally intractable to recover $x$.
For these reasons, we often seek to recover $x$ approximately by minimizing the expected loss for an appropriate loss function.
For example, if $\theta(y)$ denotes the estimate for $x$ given the measurement $y$, a natural goal would be to minimize the expected mean-square loss given by $\E_{y \sim \cD_x} [ \norm{\theta(y)-x}^2]$.

In many cases, we can formulate such a minimization problem as a feasibility problem for a system of polynomial equations.
By classical NP-completeness results, general polynomial systems in many variables are computationally intractable in the worst case.
In our context, an estimation problem gives rise to a distribution over polynomial systems that encode it.
We wish to study a typical system drawn from this distribution.
If the underlying distributions are sufficiently well-behaved, polynomial systems yield an avenue to design algorithms for high-dimensional estimation problems.

In this survey, our tool for studying such polynomial systems will be sum-of-squares (SoS) proofs.
Sum-of-squares proofs yield a complete proof system for reasoning about polynomial systems \cite{krivine1964anneaux,stengle1974nullstellensatz}.
More importantly, SoS proofs are constructive: the problem of finding a sum-of-squares proof can be formulated as a semidefinite program, and thus algorithms for convex optimization can be used to find a sum-of-squares proof when one exists.
Low-degree SoS proofs can be found efficiently, and the computational complexity of the algorithm grows exponentially with the degree of the polynomials involved in the proof.

The study of low-degree SoS proofs in the context of estimation problems suggests a rich family of questions.
For natural estimation problems, if a polynomial system drawn from the corresponding distribution is feasible, can one harness sum-of-squares proofs towards solving the polynomial system? (surprisingly, the answer is often yes!)
If a system from this distribution is typically infeasible, what is the smallest degree of a sum-of-squares refutation?
Are there structural characterizations of the degree of SoS refutations in terms of the properties of the distribution?
Is there a connection between the existence of low-degree SoS proofs and the spectra of random matrices associated with the distribution (yielding efficient spectral algorithms)?
Over the past few years, significant strides have been made on all these fronts, exposing the contours of a rich theory that remains largely hidden.
This survey will be devoted to expounding some of the major developments in this context.

\subsection{Estimation problems}
We will start by describing a few estimation problems that will be recurring examples in our survey.

\begin{example}[\kclique]\label[example]{example:kclique}
Fix a positive integer $k \leq n$.
In the \kclique problem, a clique of size $k$ is planted within a random graph drawn from the \Erdos-\Renyi distribution denoted $\gnhalf$. The goal is to recover the $k$-clique.

Formally, the structured family $\{\cD_S \}$ is parametrized by subsets $S \subset \binom{[n]}{k}$.
    For a subset $S \in \binom{[n]}{k}$, the distribution $\cD_{S}$ over measurements $G \in \bits^{\binom{n}{2}}$ is specified by the following sampling procedure:
\begin{itemize}
\item Sample a graph $G' = ([n],E(G'))$ from the \Erdos-\Renyi distribution $\gnhalf$ and set $G = ([n], E(G') \cup E(K_S))$ where $K_S$ denotes the clique on the vertices in $S$.
\end{itemize}

    An application of the second moment method \cite{grimmett_mcdiarmid_1975} shows that for all $k \gg 2 \log{n}$, the clique $S$ can be exactly recovered with high probability given the graph $G$.
    However, for any $k \ll \sqrt{n}$, there is no known polynomial time algorithm for the problem with the best algorithm being a brute force search running in time $n^{O(\log n)}$.
    Improving upon this runtime is an open problem dating back to Karp in 1976 \cite{MR0445898-Karp76}, but save for the spectral algorithm of Alon \etal\ for $k = \Omega(\sqrt{n})$ \cite{DBLP:journals/rsa/AlonKS98}, the only progress has been in proving lower bounds against broad classes of algorithms (e.g. \cite{DBLP:journals/rsa/Jerrum92,MR1969394-Feige03,DBLP:journals/jacm/FeldmanGRVX17,DBLP:conf/focs/BarakHKKMP16}).

We will now see how to encode the problem as a polynomial system.
    For pairs $(S,G)$, let $y \in \{\pm 1\}^{\binom{n}{2}}$ denote the natural $\{\pm 1\}$-encoding of the graph $G$, namely, $y_{ij} = (1-2\cdot\Ind[(i,j) \not\in E(G)])$ for all $i,j \in \binom{n}{2}$.  Set $x \seteq \Ind_S \in \bits^n$.
We will refer to the variables $y_{ij}$ as {\em instance variables} as they specify the input to the problem.
    The variables $x_i$ will be referred to as the {\em hidden variables}.
We encode each constraint as a polynomial equality or inequality:
\begin{align*}
&\text{$x_i$ are Boolean}  &     \left\{ x_i(1-x_i)  = 0\right\}_{i \in [n]}  \\
&\text{if $(i,j) \notin E(G)$ then $\{i,j\}$ are not both in clique} & \left\{(1-y_{ij})x_i x_j  = 0 \right\}_{\forall i,j \in \binom{[n]}{2}} \\
&\text{at least $k$ vertices in clique } & \sum_{i \in [n]} x_i - k \geq 0
\end{align*}
Note that when we are solving the estimation problem, the instance variables $y_{ij}$ are given, and the hidden variables $\{x_i\}$ are the unknowns in the polynomial system.
It is easy to check that the only feasible solutions $x \in \R^n$ for this system of polynomial equations are Boolean vectors $x \in \{0,1\}^n$ which are supported on cliques of size at least $k$ in $G$.
\end{example}

\paragraph{Refutation and distinguishing.} For every estimation problem that we will encounter in this survey, we can associate two related computational problems termed {\em refutation} and {\em distinguishing}.
In estimation problems, we typically think of instances $y$ as having structure: we sample $y$ from a structured distribution $\cD_x$, and we wish to recover the hidden variables $x$ that give structure to $\cD_x$.
But there may also be instances $y$ which do not have structure.
The goal of {\em refutation} is to certify that there is no hidden structure, when there is none.

A \emph{null} distribution is a probability distribution over instances $y$ for which there is no hidden structure $x$.
For example, in the \kclique problem, the corresponding null distribution is the \Erdos-\Renyi random graph $\gnhalf$ (without a planted clique).
With high probability, a graph $y \sim \gnhalf$ has no clique with significantly more than $2\log{n}$ vertices.
Therefore, for a fixed $k \gg 2 \log n$, given a graph $y \sim \gnhalf$, the goal of a refutation algorithm is to certify that $y$ has no clique of size $k$.
Equivalently, the goal of a refutation algorithm is to certify the infeasibility of the associated polynomial system.

The most rudimentary computational task associated with estimation and refutation is that of {\em distinguishing}.
The setup of the distinguishing problem is as follows.
Fix a prior distribution $\pi$ on the hidden variables $x \in \R^n$, which in turn induces a distribution $\planteddist$ on $\R^N$, obtained by first sampling $x \sim \pi$ and then sampling $y \sim \cD_x$.
The input consists of a sample $y$ which is with equal probability drawn from the structured distribution $\planteddist$ or the null distribution $\randomdist$.
The computational task is to identify which distribution the sample $y$ is drawn from, with a probability of success $\frac{1}{2}+\delta$ for some constant $\delta > 0$.
For example, the structured distribution for \kclique is obtained by setting the prior distribution of $x$ to be uniform on subsets of $[n]$ of size $k$.
In the distinguishing problem, the input is a graph drawn from either $\planteddist$ or the null distribution $\gnhalf$, and the algorithm is required to identify the distribution.
For every problem included in this survey, the distinguishing task is formally no harder than estimation or refutation, i.e., the existence of algorithms for estimation or refutation immediately implies a distinguishing algorithm.

\begin{example} (\tensorpca) \label[example]{ex:tensorpca}
The family of structured distributions $\{\cD_v\}$ is parametrized by unit vectors $v \in \R^n$.
    A sample from $\cD_v$ consists of a $4$-tensor $\bT = v^{\otimes 4} + \zeta$
where $\zeta \in \R^{n
\times n \times n \times n}$ is a symmetric $4$-tensor whose entries are i.i.d Gaussian random variables sampled from $N(0,\sigma^2)$.  The goal is to recover a vector $x$ that is close as possible to $v$.
\end{example}

A canonical strategy to recover $v$ given $\bT = v^{\otimes 4} + \zeta$ is to maximize the degree-$4$ polynomial associated with the symmetric $4$ tensor $\bT$.
Specifically, if we set
\[x' = \argmax_{\norm{x} \leq 1} \iprod{\bT,x^{\otimes 4}} \]
then one can show that $\norm{v - x'}_2 \leq O(n^{1/2} \cdot \sigma)  $ with high probability over $\zeta$.
If $\bT \sim \cD_v$ then $ \iprod{\bT, v^{\otimes 4} } = 1$.
Furthermore, when $\sigma \ll n^{-1/2}$ it can be shown that $v \in \R^n$ is close to the unique maximizer of the function $\phi(x) = \iprod{\bT,x^{\otimes 4}}$.
So the problem of recovering $v$ can be encoded as the following polynomial system:
\begin{align*}
    &\text{$x$ is in the unit sphere}   &\|x\|^2  \leq 1, \\
    & \text{$x$ has large value for $\bT$} &  \sum_{i,j,k, \ell \in [n]^4} \bT_{ijk \ell } x_i x_j x_k x_\ell \geq 1.
\end{align*}

In the distinguishing and refutation versions of this problem, we will take the \emph{null} distribution $\randomdist$ to be the distribution over $4$-tensors with independent Gaussian entries sampled from $N(0,\sigma^2)$ (equivalent to the distribution of the noise $\zeta$ from $\cD_v$).
For a $4$-tensor $\bT$, the maximum of $\bT(x) = \iprod{x^{\otimes 4}, \bT}$ over the unit ball is referred to as the \emph{injective tensor norm}  of the tensor $\bT$, and is denoted by $\norminj{\bT}$.
If $\bT \sim \randomdist$ then $\norminj{\bT} \leq O \left(n^{1/2} \cdot \sigma \right)$  with high probability over choice of $\bT$ \cite{doi:10.1002/cpa.21422}.
Thus when $\sigma \ll n^{-1/2}$, the refutation version of the \tensorpca problem reduces to certifying an upper bound on $\norminj{\bT}$.
If we could compute $\norminj{\bT}$ exactly, then we can certify that $\bT \sim \randomdist$ for $\sigma$ as large as $\sigma = O(n^{-1/2})$.

The injective tensor norm is known to be computationally intractable in the worst case \cite{gurvits2003classical, gharibian2010strong, DBLP:conf/stoc/BarakBHKSZ12}.
Understanding the the function $\iprod{x^{\otimes k},\bT}$ for random $\bT$ is a deep topic in probability theory and statistical physics (e.g. \cite{doi:10.1002/cpa.21422}).
As an estimation problem, \tensorpca was first considered by \cite{MontanariR14}, and inspired multiple follow-up works concerned with spectral and SoS algorithms (e.g. \cite{DBLP:conf/colt/HopkinsSS15,DBLP:conf/stoc/HopkinsSSS16,DBLP:conf/stoc/RaghavendraRS17,DBLP:conf/approx/BhattiproluGL17}).

\begin{example}(Matrix \& Tensor Completion)
In matrix completion, the hidden parameter is a rank-$r$ matrix $X \in \R^{n \times n}$.
For a parameter $X$, the measurement consists of a partial matrix revealing a subset of entries of $X$, namely $X_{\Omega}$ for a subset $\Omega \subset [n] \times [n]$ with $\card{\Omega} = m$.
The probability distribution $\cD_X$ over measurements is obtained by picking the set $\Omega$ to be a uniformly random subset of $m$ entries.

To formulate a polynomial system for recovering a rank-$r$ matrix consistent with the measurement $X_{\Omega}$, we will use a $n \times r$ matrix of variables $B$, and write the following system of constraints on it:
\begin{align*}
    & \text{$BB^T$ is consistent with measurement} & (BB^T)_{\Omega} = X_{\Omega}\mper&
\end{align*}
Tensor completion is the analogous problem with $\bX$ being a higher-order tensor namely, $\bX = \sum_{i = 1}^r a_i^{\otimes k}$ for some fixed $k \in \N$.
The corresponding polynomial system is again over a $n \times r$ matrix of variables $B$ with columns $b_1,\ldots,b_r$ and the following system of constraints,
\begin{align*}
& \text{$\textstyle \sum_{i=1}^r b_i^{\otimes k}$ is consistent with measurement}
& \left(\sum_{i \in [r]} b_i^{\otimes k}\right)_{\Omega} = \bX_{\Omega}\mper
\end{align*}
\end{example}

\subsection{Sum-of-squares proofs}\label{sec:sos}
The sum-of-squares (SoS) proof system is a restricted class of proofs for reasoning about polynomial systems.
Fix a set of polynomial inequalities $\cA = \{ p_i(x) \geq 0 \}_{i \in [m]}$ in variables $x_1,\ldots,x_n$.
We will refer to these inequalities as the \emph{axioms}.
Starting with the axioms $\cA$, a sum-of-squares proof of $q(x) \geq 0$ is given by an identity of the form,
\[
    q(x) = \sum_{j\in [k]} s_j^2(x) + \sum_{i \in [m]} a_i^2(x) \cdot p_i(x) \mcom
\]
where $\{s_j(x)\}_{j \in [k]},\{a_i(x)\}_{i \in [m]}$ are real polynomials.
It is clear that any identity of the above form manifestly certifies that the polynomial $q(x) \geq 0$, whenever each $p_i(x) \geq 0$ for real $x$.
The degree of the sum-of-squares proof is the maximum degree of all the summands, i.e., $\max \{\deg(s_j^2), \deg(a_i^2 p_i)\}_{i,j}$.

Sum-of-squares proofs extend naturally to polynomial systems that involve a set of equalities $\{r_i(x) = 0\}$ along with a set of inequalities $\{ p_i(x) \geq 0\}$.
We can extend the definition syntactically by replacing each equality $r_i(x) = 0$ by a pair of inequalities $r_i(x) \geq 0$
and $-r_i(x) \geq 0$.
\Pnote{}

We will the use the notation $\cA \sststile{d}{x} \set{q(x) \geq 0}$ to denote that the assertion that,  there exists a degree-$d$ sum-of-squares proof of $q(x) \geq 0$ from the set of axioms $\cA$.
The superscript $x$ in the notation $\cA \sststile{d}{x} \set{q(x) \geq 0}$ indicates that the sum-of-squares proof is an identity of polynomials where $x$ is the formal variable.
We will drop the subscript or superscript when it is clear from the context, and just write $\cA \sststile{}{} \set{q(x) \geq 0}$.
Sum-of-squares proofs can also be used to certify the infeasibility, or \emph{refute}, the polynomial system.
In particular, a degree-$d$ sum-of-squares refutation of a polynomial system $\{ p_i(x) \geq 0\}_{i \in [m]}$ is an identity of the form,
\begin{align}
-1 &= \sum_{i \in [k]} s_i^2(x) + \sum_{i \in [m]} a^2_i(x) \cdot p_i(x) \label{eq:refute}
\end{align}
where $\max \{\deg(s_j^2), \deg(a_i^2 p_i)\}_{i,j}$ is at most $d$.

The sum-of-squares proof system has been an object of study starting with the work of Hilbert and Minkowski more than a century ago (see \cite{MR1747589-Reznick00} for a survey).
With no restriction on degree, Stengle's Positivestellensatz implies that sum-of-squares proofs form a complete proof system, i.e., if the axioms $\cA$ imply $q(x) \geq 0$, then there is an SoS proof of this fact.

The algorithmic implications of the sum-of-squares proof system were first realized in the works of Parrilo \cite{parrilo2000structured} and Lasserre \cite{MR1814045-Lasserre00}, who independently arrived at families of algorithms for polynomial optimization using semidefinite programming (SDP).
Specifically, these works observed that semidefinite programming can be used to find a degree-$d$ SoS proof in time $n^{O(d)}$, if there exists one.
This family of algorithms (called a hierarchy, as we have an algorithm for each even integer degree-$d$) are referred to as the sum-of-squares SDP hierarchy.
We say that the SoS algorithm is {\em low-degree} if $d$ does not grow with $n$.

The SoS hierarchy has since emerged as a powerful tool for algorithm design.
On the one hand, the first few levels of the SoS hierarchy systematically capture a vast majority of algorithms in combinatorial optimization and approximation algorithms developed over several decades.
Furthermore, the low-degree SoS SDP hierarchy holds the promise of yielding improved approximations to NP-hard combinatorial optimization problems, approximations that would beat the long-standing and universal barrier posed by the notorious unique games conjecture \cite{MR2869009-Trevisan12,DBLP:journals/eccc/BarakS14}.

More recently, the low-degree SoS SDP hierarchy has proved to be a very useful tool in designing algorithms for high-dimensional estimation problems, wherein the inputs are drawn from a natural probability distribution.
For this survey, we organize the recent work on this topic into three lines of work.
\begin{itemize}
\item \emph{When the polynomial system for an estimation problem is feasible, can sum-of-squares proofs be harnessed to retrieve the solution? }
    The answer is {\bf yes} for many estimation problems, including tensor decomposition, matrix and tensor completion, and clustering problems.
Furthermore, there is a simple and unifying principle that underlies all of these applications.
Specifically, the underlying principle asserts that if there is a low-degree SoS proof that all solutions to the system are close to the hidden variable $x$, then a low-degree SoS SDP can be used to actually retrieve $x$.
We will discuss this broad principle and several of its implications in \cref{sec:algorithms}.

\item \emph{When the polynomial system is infeasible, what is the smallest degree at which it admits sum-of-squares proof of infeasibility?}
The degree of the sum-of-squares refutation is critical for the run-time of the SoS SDP-based algorithm.
Recent work by Barak \etal \cite{DBLP:conf/focs/BarakHKKMP16} introduces a technique referred to as ``pseudocalibration'' for proving lower bounds on the degree of SoS refutation, developed in the context of the work on \kclique.
\cref{sec:lowerbounds} is devoted to the heuristic technique of pseudocalibration, and the mystery surrounding its effectiveness.

\item \emph{Can the existence of degree-$d$ of sum-of-square refutations be characterized in terms of (spectral) properties of the underlying distribution?}
In \cref{sec:spectral}, we will discuss a result that shows a connection between the existence of low-degree sum-of-squares refutations and the spectra of certain low-degree matrices associated with the distribution.
This connection implies that under fairly mild conditions, the SoS SDP based algorithms are no more powerful than a much simpler and more lightweight class of algorithms referred to as \emph{spectral algorithms}.
Roughly speaking, a spectral algorithm proceeds by constructing a matrix $M(x)$ out of the input instance $x$, and then using the eigenvalues of the matrix $M(x)$ to recover the desired outcome.
\end{itemize}

%% file: content/prelims.tex
\paragraph{Notation.}
For a positive integer $n$, we use $[n]$ to denote the set $\{1,\ldots,n\}$.
We sometimes use $\binom{[n]}{d}$ to denote the set of all subsets of $[n]$ of size $d$, and $[n]^{\le d}$ to denote the set of all multi-subsets of cardinality at most $d$.

If $x \in \R^n$ and $A \subset [n]$ is a multiset, then we will use the shorthand $x^A$ to denote the monomial $x^A = \prod_{i \in A} x_i$.
We will also use $x^{\le d}$ to denote the $N \times 1$ vector containing all monomials in $x$ of degree at most $d$ (including the constant monomial $1$), where $N = \sum_{i=0}^d n^i$.
Let $\R[x]_{\leq d}$ denote the space of polynomials of degree at most $d$ in variables $x$.

For a function $f(n)$, we will say $g(n) = O(f(n))$ if $\lim_{n\to\infty}\frac{g(n)}{f(n)} \le C$ for some universal constant $C$.
We say that $f(n) \ll g(n)$ if $\lim_{n\to\infty}\frac{f(n)}{g(n)}= 0$.

If $\mu$ is a distribution over the probability space $\cS$, then we use the notation $x \sim \mu$ for $x \in \cS$ sampled according to $\mu$.
For an event $\cE$, we will use $\Ind[\cE]$ as the indicator that $\cE$ occurs.
We use $\gnhalf$ to denote the \Erdos-\Renyi distribution with parameter $\tfrac{1}{2}$, or the distribution over graphs where each edge is included independently with probability $\tfrac{1}{2}$.

If $M$ is an $n \times m$ matrix, we use $\lmax(M)$ to denote $M$'s largest eigenvalue.
When $n=m$, then $\Tr(M)$ denotes $M$'s trace.
If $N$ is an $n \times m$ matrix as well, then we use $\iprod{M,N} = \Tr(MN^\top)$ to denote the {\em matrix inner product}.
We use $\normf{M}$ to denote the Frobenius norm of $M$, $\normf{M} = \iprod{M,M}^{1/2}$.
For a subset $S \subset [n]$, we will use $\Ind_S$ to denote the  $\bits$ indicator vector of $S$ in $\R^n$.
We will also use $\Ind$ to denote the all-1's vector.

For two matrices $A,B$ we use $A \otimes B$ to denote both the Kronecker product of $A$ and $B$, and the order-$4$ tensor given by taking $A \otimes B$ and reshaping it with modes for the rows and columns of $A$ and of $B$.
We also use $A^{\otimes k}$ to denote the $k$-th Kronecker power of $A$.
For an order-$k$ tensor $T \in (\R^n)^{\otimes k}$ and for a permutation of $[k]$ $i_1,\ldots,i_k$, we denote by $T_{\{i_1,\ldots,i_\ell\}\{i_{\ell+1},\ldots,i_{k}\}}$ the $n^{\ell} \times n^{k-\ell}$ matrix reshaping given by ordering the modes of $T$ so that $i_1,\ldots,i_{\ell}$ index the rows and $i_{\ell+1},\ldots,i_k$ index the columns.

\paragraph{Pseudoexpectations.}
For a polynomial system $\cP$ in $n$ variables $x \in \R^n$ consisting of inequalities $\{p_i(x) \ge 0\}_{i\in[m]}$, we can write an SDP of size $n^{O(d)}$ which finds a degree-$d$ sum-of-squares refutation, if one exists (see \cite{rothvoss2013lasserre} for more discussion).

If there is no degree-$d$ refutation, the dual semidefinite program computes in time $n^{O(d)}$ a linear functional over degree-$d$ polynomials which we term a {\em pseudoexpectation}.
Formally, a degree-$d$ pseudoexpectation $\pE : \R[x]_{\leq d} \to \R$ is a linear functional over polynomials of degree at most $d$ with the properties that $\pE[1] = 1$, $\pE[p(x) a^2(x)] \geq 0$ for all $p \in \cP$ and polynomials $a$ such that $\deg(a^2 \cdot p) \le d$, and $\pE[q(x)^2] \ge 0$ whenever $\deg(q^2) \le d$.

\begin{claim} \label{claim:functionalNoRefutation}
If there exists a degree-$d$ pseudoexpectation $\pE : \R[x]_{\leq d} \to \R$ for the polynomial system $\cP = \{ p_i(x) \geq 0 \}_{i \in [m]}$, then $\cP$ does not admit a degree-$d$ refutation.
\end{claim}
\begin{proof}
	Suppose $\cP$ admits a degree-$d$ refutation.
    Applying the pseudoexpectation operator $\pE$ to the left-hand-side of \cref{eq:refute}, we have $-1$.
    Applying $\pE$ to the right-hand-side of \cref{eq:refute}, the first summand must be non-negative by definition of $\pE$ since it is a sum of squares, and the second summand is non-negative, since we assumed that $\pE$ satisfies the constraints of $\cP$.
    This yields a contradiction.
\end{proof}

The properties above imply that when $\cA \sststile{d}{x} \set{q(x) \ge 0}$, then if $\pE$ is a degree-$d$ pseudoexpectation operator for the polynomial system defined by $\cA$, $\pE[q(x)] \ge 0$ as well.
This implies that $\pE$ satisfies several useful inequalities; for example, the Cauchy-Schwarz inequality.
\begin{claim}\label[claim]{cl:sos-cs}
    If $\pE$ is a degree-$d$ pseudoexpectation and if $p,q$ are polynomials of degree at most $\frac{d}{2}$, then $\pE[q(x)\cdot p(x)] \le \frac{1}{2}\pE[q(x)^2] + \frac{1}{2}\pE[p(x)^2]$.
\end{claim}
\begin{proof}
We have the following polynomial equality of degree at most $d$:
    \[
	q(x) p(x) = \frac{1}{2}\cdot q(x)^2 + \frac{1}{2}\cdot q(x)^2 - \frac{1}{2}\left(q(x) - p(x)\right)^2.
	\]
	Applying $\pE$ to both sides, using that $\pE[(q(x) - p(x))^2] \ge 0$, we have our conclusion.
\end{proof}
Other versions of the Cauchy-Schwarz inequality can be shown to hold for pseudoexpectations as well; see e.g. \cite{DBLP:conf/stoc/BarakBHKSZ12} for details.

%% file: content/algorithms.tex
\section{Algorithms for high-dimensional estimation}\label[section]{sec:algorithms}

In this section, we prove a algorithmic meta-theorem for high-dimensional estimation that provides a unified perspective on the best known algorithms for a wide range of estimation problems.
This unifying perspective allows us to obtain algorithms with significantly better guarantees than what's known to be possible with other methods.
We illustrate the power of this meta-theorem by applying it to matrix and tensor completion, tensor decomposition, and clustering.

\subsection{Algorithmic meta-theorem for estimation}

We consider the following general class of estimation problems, which will turn out to capture a plethora of interesting problems in a useful way:
In this class, an estimation problem\footnotemark is specified by a set $\cP\subseteq \R^n\times \R^m$ of pairs $(x,y)$, where $x$ is called \emph{parameter} and $y$ is called \emph{measurement}.
Nature chooses a pair $(x^*,y^*)\in \cP$, we are given the measurement $y^*$ and our goal is to (approximately) recover the parameter $x^*$.

\footnotetext{
  In contrast to the discussion of estimation problems in \cref{sec:introduction}, for every parameter, we have a set of possible measurements as opposed to a distribution over measurements.
  We can model distributions over measurements in this way by considering a set of ``typical measurements''.
  The viewpoint in terms of sets of possible measurements will correspond more closely to the kind of algorithms we consider.
}

For example, we can encode compressed sensing with measurement matrix $A\in \R^{m\times n}$ and sparsity bound $k$ by the following set of pairs,
\begin{displaymath}
  \cP_{A,k} = \Set{(x,y) \mid y= Ax,~ \text{$x\in\R^n$ is $k$-sparse}}\,.
\end{displaymath}
Similarly, we can encode matrix completion with observed entries $\Omega\subseteq [n]\times [n]$ and rank bound $r$ by the set of pairs,
\begin{displaymath}
  \cP_{\Omega,r} = \Set{(X,X_\Omega) \mid X\in \R^{n\times n},~\rank X \le r}\,.
\end{displaymath}
For both examples, the measurement was a simple (linear) function of the parameter.
This is not always the case; consider for example the following clustering problem.
There are two distinct centers $\mu_1,\mu_2 \in \R^n$, and we observe $m$ samples $y_1,\ldots y_m \in \R^n$ such that each sample is closer to either $\mu_1$ or $\mu_2$.
Then we can encode the problem of finding $\mu_1$ and $\mu_2$ as follows,
\begin{displaymath}
    \cP_{\mu,m} = \Set{ ((\mu_1,\mu_2),Y) \mid \mu_1\neq \mu_2 \in \R^n, ~Y \in \R^{n\times m}, ~\forall i \in [m],~ \left|\|y_i - \mu_1\| - \|y_i - \mu_2\|\right| > 0}.
\end{displaymath}
\Tnote{}

\paragraph{Identifiability.}

In general, an estimation problem $\cP\subseteq \R^{n}\times \R^m$ may be ill-posed in the sense that, even ignoring computational efficiency, it may not be possible to (approximately) recover the parameter for a measurement $y$ because we have $(x,y),(x',y)\in\cP$ for two far-apart parameters $x$ and $x'$.

For a pair $(x,y)\in \cP$, we say that \emph{$y$ identifies $x$ exactly} if $(x',y)\not\in \cP$ for all $x'\neq x$.
Similarly, we say that \emph{$y$ identifies $x$ up to error $\e>0$} if $\norm{x-x'}\le \e$ for all $(x',y)\in \cP$.
We say that $x$ is identifiable (up to error $\e$) if every $(x,y)\in \cP$ satisfies that $y$ identifies $x$ (up to error $\e$).

For example, for compressed sensing $\cP_{A,k}$, it is not difficult to see that every $k$-sparse vector is identifiable if every subset of at most $2k$ columns of $A$ is linearly independent.
For tensor decomposition, a sufficient condition under which the observation $f(x)=\sum_{i=1}^r x_i^{\otimes 3}$ is enough to identify $x\in\R^{n\times r}$ (up to a permutation of its columns) is if the columns $x_1,\ldots,x_r\in\R^n$ of $x$ are linearly independent.

\paragraph{From identifiability proofs to efficient algorithms.}

By itself, identifiability typically only implies that there exists an inefficient algorithm to recover a vector $x$ close to the parameter $x^*$ from the observation $y^*$ (e.g. by brute-force search over the set of all $x$).
But perhaps surprisingly, the notion of identifiability in a broader sense can also help us understand if there exists an efficient algorithm for this task.
Concretely, if the \emph{proof of identifiability} is captured by the sum-of-squares proof system at low degree, then there exists an efficient algorithm to (approximately) recover $x$ from $y$.

In order to formalize this phenomenon, let the set $\cP\subseteq \R^n\times \R^m$ be be described by polynomial equations
\begin{displaymath}
  \cP=\Set{(x,y) \mid \exists z.~p(x,y,z)=0}\,,
\end{displaymath}
where $p=(p_1,\ldots,p_t)$ is a vector-valued polynomial and $z$ are auxiliary variables.\footnote{
  We allow auxiliary variables here because they might make it easier to describe the set $\cP$.
  The algorithms we consider depend on the algebraic description of $\cP$ we choose and different descriptions can lead to different algorithmic guarantees.
  In general, it is not clear which description is best.
  However, typically, the more auxiliary variables the better.
}
(In other words, $\cP$ is a projection of the variety given by the polynomials $p_1,\ldots,p_t$.)
\Tnote{}
The following theorem shows that there is an efficient algorithm to (approximately) recover $x^*$ given $y^*$ if there exists a low-degree proof of the fact that the equation $p(x,y^*,z)=0$ implies that $x$ is (close to) $x^*$.

\begin{theorem}[Meta-theorem for efficient estimation]
  \label[theorem]{deterministic-meta-theorem}
  Let $p$ be a vector-valued polynomial and let the triples $(x^*,y^*,z^*)$ satisfy $p(x^*,y^*,z^*)=0$.
  Suppose $\cA \sststile{d}{x,z} \set{\norm{x^*-x}^2\le \e}$, where $\cA=\set{p(x,y^*,z)=0}$.
  Then, every degree-$d$ pseudo-distribution $D$ consistent with the constraints $\cA$ satisfies
  \begin{displaymath}
    \Norm{x^*-\pE_{D(x,z)}x}^2 \le \e\,.
  \end{displaymath}
  Furthermore, for every $d\in \N$, there exists a polynomial-time algorithm (with running time $n^{O(d)}$)\footnotemark{} that given a vector-valued polynomial $p$ and a vector $y$ outputs a vector $\hat x(y)$ with the following guarantee:
  if $\cA \sststile{d}{x,z} \set{\norm{x^*-x}^2\le \e}$ with a proof of bit-complexity at most $n^{d}$, then $\norm{x^*-\hat x(y^*)}^2 \le \e+2^{-n^{d}}$.
\end{theorem}

Despite not being explicitly stated, the above theorem is the basis for many recent advances in algorithms for estimation problems through the sum-of-squares method \cite{DBLP:conf/stoc/BarakKS15,DBLP:conf/stoc/BarakKS14,DBLP:conf/colt/HopkinsSS15,DBLP:conf/focs/MaSS16,DBLP:conf/colt/BarakM16,DBLP:conf/colt/PotechinS17,KothariSS18,HopkinsL18}.

\footnotetext{In order to be able to state running times in a simple way, we assume that the total bit-complexity of $(x,y,z)$ and the vector-valued polynomial $p$ (in the monomial basis) is bounded by a fixed polynomial in $n$.}

\begin{proof}
  Let $D$ be a degree-$d$ pseudo-distribution $D$ with $D\sststile{}{x,z} \cA$.
  Since degree-$d$ sum-of-squares proofs are sound for degree-$d$ pseudo-distributions, we have $D\sststile{d}{x,z}\set{\norm{x^*-x}^2 \le \e}$.
  In particular, $\pE_{D(x,z)} \norm{x^*-x}^2\le \e$.
    By Cauchy--Schwarz for pseudo-distributions (\cref{cl:sos-cs}), every vector $u\in \R^n$ satisfies
  \begin{displaymath}
    \iprod{u, \pE_{D(x,z)} x^*-x} = \pE_{D(x,z)}\iprod{u,x^*-x}\le \Paren{\pE \norm{u}^2}^{1/2}\cdot \Paren{\pE \norm{x^*-x}^2}^{1/2} \le \norm{u} \cdot \e^{1/2}\,.
  \end{displaymath}
  By choosing $u=\pE_{D(x,z)} x^*-x$, we obtain the desired conclusion about $\pE_{D(x,z)}x$.

  Given a measurement $y^*$, the algorithm computes a degree-$d$ pseudo-distribution $D(x,z)$ that satisfies $\cA$ up to error $2^{-n^{100d}}$ and outputs $\hat x(y^*)=\pE_{D(x,z)} x$.
  \Dnote{}
  We are guaranteed that such a pseudo-distribution exists, e.g. the distribution that places all its probability mass on the vector $x^*$.
    If the proof $\cA \sststile{d}{x} \set{\norm{x^*-x}^2\le \e}$ has bit-complexity $n^{d}$, it follows that $D(x)$ satisfies $\set{\norm{x^*-x}^2\le \e}$ up to error $2^{-n^d}$. \Tnote{}
  In particular, $\pE_{D(x)}\norm{x^*-x}^2\le \e + 2^{-n^d}$.
  By the same argument as before, it follows that $\norm{x^*-\hat x(y^*)}^2\le \e + 2^{-n^d}$.
\end{proof}

\subsection{Matrix and tensor completion}

In matrix completion, we observe a few entries of a low-rank matrix and the goal is to fill in the missing entries.
This problem is studied extensively both from practical and theoretical perspectives.
One of its practical applications is in recommender systems, which was the basis of the famous Netflix Prize competition.
Here, we may observe a few movie ratings for each user and the goal is to infer a user's preferences for movies that the user hasn't rated yet.

In terms of provable guarantees, the best known polynomial time algorithm for matrix completion is based on a semidefinite programming relaxation.
Let $X=\sum_{i=1}^r \sigma_i \cdot u_i \transpose {v_i}\in\R^{n\times n}$ be a rank-$r$ matrix such that its left and right singular vectors $u_1,\ldots,u_r,v_1,\ldots,v_r\in\R^n$ are $\mu$-incoherent\footnotemark{}, i.e., they satisfy $\iprod{u_i,e_j}^2\le \mu/n$ and $\iprod{v_i,e_j}^2\le \mu/n$ for all $i\in [r]$ and $j\in[n]$.
The algorithm observes the partial matrix $X_\Omega$ that contains a random cardinality $m$ subset $\Omega\subseteq [n]\times [n]$ of the entries of $X$.
If $m\ge \mu r n \cdot O(\log n)^2$, then with high probability over the choice of $\Omega$ the algorithm recovers $X$ exactly \cite{DBLP:journals/focm/CandesR09,DBLP:journals/tit/Gross11,DBLP:journals/jmlr/Recht11,DBLP:journals/tit/Chen15}.
This bound on $m$ is nearly optimal in that $m\ge \Omega(r n)$ appears to be necessary because an $n$-by-$n$ rank-$r$ matrix has $\Omega(r n)$ degrees of freedom (the entries of its singular vectors).

\footnotetext{
  Random unit vectors satisfy this notion of $\mu$-incoherence for $\mu\le O(\log n)$.
  In this sense, incoherent vectors behave similar to random vectors.
}

In this section, we will show how the above algorithm is captured by sum-of-squares and, in particular, \cref{deterministic-meta-theorem}.
We remark that this fact follows directly by inspecting the analysis of the original algorithm \cite{DBLP:journals/focm/CandesR09,DBLP:journals/tit/Gross11,DBLP:journals/jmlr/Recht11,DBLP:journals/tit/Chen15}.
The advantage of sum-of-squares here is two-fold:
First, it provides a unified perspective on algorithms for matrix completion and other estimation problems.
Second, the sum-of-squares approach for matrix completion extends in a natural way to tensor completion (in a way that the original approach for matrix completion does not).

\paragraph{Identifiability proof for matrix completion.}
For the sake of clarity, we consider a simplified setup where the matrix $X$ is assumed to be a rank-$r$ projector so that $X=\sum_{i=1}^r \dyad {a_i}$ for $\mu$-incoherent orthonormal vectors $a_1,\ldots,a_r\in\R^n$.
The following theorem shows that, with high probability over the choice of $\Omega$, the matrix $X$ is identified by the partial matrix $X_\Omega$.
Furthermore, the proof of this fact is captured by sum-of-squares.
Together with \cref{deterministic-meta-theorem}, the following theorem implies that there exists a polynomial-time algorithm to recover $X$ from $X_\Omega$.

\begin{theorem}[implicit in \cite{DBLP:journals/focm/CandesR09,DBLP:journals/tit/Gross11,DBLP:journals/jmlr/Recht11,DBLP:journals/tit/Chen15}]
  \label[theorem]{matrix-completion-identifiability}
  Let $X=\sum_{i=1}^r \dyad {a_i}\in\R^{n\times n}$ be an $r$-dimensional projector and $a_1,\ldots,a_r\in\R^n$ orthonormal with incoherence $\mu=\max_{i,j} n \cdot \iprod{a_i,e_j}^2$.
  Let $\Omega\subseteq [n]\times [n]$ be a random symmetric subset of size $\card{\Omega}=m$.
  Consider the system of polynomial equations in the $n$-by-$r$ matrix variable $B$,
  \begin{displaymath}
    \cA =\Set{  \Paren{\dyad B}_\Omega = X_\Omega ,~ \transpose B B = \Id_r}\,.
  \end{displaymath}
  Suppose $m\ge \mu r n\cdot O(\log n)^2$.
  Then, with high probability over the choice of $\Omega$,
  \begin{displaymath}
    \cA \sststile{4}{B} \Set{\normf{\dyad B - X}=0}\,.
  \end{displaymath}
\end{theorem}

\begin{proof}
  The analyses of the aforementioned algorithm for matrix completion \cite{DBLP:journals/focm/CandesR09,DBLP:journals/tit/Gross11,DBLP:journals/jmlr/Recht11,DBLP:journals/tit/Chen15} show the following:
    let $\overline \Omega$ be the complement of $\Omega$ in $[n]\times[n]$.
    Then if $X$ satisfies our incoherence assumptions, with high probability over the choice of $\Omega$, there exists\footnotemark{} a symmetric matrix $M$ with $M_{\overline \Omega}=0$ and $-0.9(\Id_n -X)\preceq M-X\preceq 0.9(\Id_n -X)$.
  As we will see, this matrix also implies that the above proof of identifiability exists.

  \footnotetext{
    Current proofs of the existence of this matrix proceed by an ingenious iterative construction of this matrix (alternatingly projecting to two affine subspaces).
    The analysis of this iterative construction is based on matrix concentration bounds.
    We refer to prior literature for details of this proof \cite{DBLP:journals/tit/Gross11,DBLP:journals/jmlr/Recht11,DBLP:journals/tit/Chen15}.
  }

  Since $0\preceq X$ and $X - 0.9 (\Id_n-X)\preceq M$, we have
  \begin{displaymath}
    \iprod{M,X}\ge \iprod{X,X}-0.9\iprod{\Id_n-X,X}=\iprod{X,X}=r\,.
  \end{displaymath}
  Since $M_{\overline \Omega}=0$ and $\cA$ contains the equation $(\dyad B)_\Omega=X_\Omega$, we have $\cA \sststile{}{B}\iprod{M,\dyad B}=\iprod{M,X}\ge r$.
  At the same time, we have
  \begin{displaymath}
    \cA \sststile{}{} \iprod{M,\dyad B}
    \le \iprod{X,\dyad B}+0.9\iprod{\Id_n-X,\dyad B}
    = 0.1\iprod{X,\dyad B} + 0.9 r\,,
    \label{matrix-completion-proof}
  \end{displaymath}
  where the first step uses $M\preceq X + 0.9(\Id -X)$ and the second step uses $\cA \sststile{}{} \iprod{\Id_n , \dyad B}=r$ because $\iprod{\Id_n , \dyad B}=\Tr \transpose B B$ and $\cA$ contains the equation $\transpose B B = \Id_r$.
  Combining the lower and upper bound on $\iprod{M,\dyad B}$, we obtain
  \begin{displaymath}
    \cA \sststile{}{}
    \iprod{X,\dyad B} \ge r\,.
  \end{displaymath}
  Together with the facts $\normf{X}^2=r$ and $\cA \sststile{}{} \normf{\dyad B}^2=r$, we obtain $\cA \sststile{}{} \normf{X-\dyad B}^2=0$ as desired.
\end{proof}

\paragraph{Identifiability proof for tensor completion.}

Tensor completion is the analog of matrix completion for tensors.
We observe a few of the entries of an unknown low-rank tensor and the goal is to fill in the missing entries.
In terms of provable guarantees, the best known polynomial-time algorithms are based on sum-of-squares, both for exact recovery \cite{DBLP:conf/colt/PotechinS17} (of tensors with orthogonal low-rank decompositions) and approximate recovery \cite{DBLP:conf/colt/BarakM16} (of tensors with general low-rank decompositions).

Unlike for matrix completion, there appears to be a big gap between the number of observed entries required by efficient and inefficient algorithms.
For 3-tensors, all known efficient algorithms require $r \cdot \tilde O(n^{1.5})$ observed entries (ignoring the dependence on incoherence) whereas information-theoretically $r\cdot O(n)$ observed entries are enough.
The gap for higher-order tensors becomes even larger.
It is an interesting open question to close this gap or give formal evidence that the gap is inherent.

As for matrix completion, we consider the simplified setup that the unknown tensor has the form $X=\sum_{i=1}^r a_i^{\otimes 3}$ for incoherent, orthonormal vectors $a_1,\ldots,a_r\in \R^n$.
The following theorem shows that with high probability, $X$ is identifiable from $r n^{1.5}\cdot (\mu \log n)^{O(1)}$ random entries of $X$ and this fact has a low-degree sum-of-squares proof.

\begin{theorem}[\cite{DBLP:conf/colt/PotechinS17}]\torestate{\label{tensor-comp}
  Let $a_1,\ldots,a_r\in\R^n$ be orthonormal vectors with incoherence $\mu=\max_{i,j} n\cdot\iprod{a_i,e_j}^2$ and let $\bX=\sum_{i=1}^r a_i^{\otimes 3}$ be their 3-tensor.
  Let $\Omega\subseteq [n]^3$ be a random symmetric subset of size $\card{\Omega}=m$.
  Consider the system of polynomial equations in the $n$-by-$r$ matrix variable $B$ with columns $b_1,\ldots,b_r$,
  \begin{displaymath}
    \cA =\Set{  \Paren{\sum_{i=1}^r b_i^{\otimes 3}}_\Omega = \bX_\Omega ,~ \transpose B B = \Id_r}
  \end{displaymath}
  Suppose $m\ge r n^{1.5} \cdot (\mu \log n)^{O(1)}$.
  Then, with high probability over the choice of $\Omega$,
  \begin{displaymath}
    \cA \sststile{O(1)}{B} \Set{\normf{\sum_{i=1}^r b_i^{\otimes 3} - \bX}^2=0}
  \end{displaymath}}
\end{theorem}

\begin{proof}
  Let $A\in\R^{n\times r}$ be the matrix with columns $a_1,\ldots,a_r$.
  Analogous to the proof for matrix completion, the heart of the proof is the existence of a 3-tensor $\bT$ that satisfies the following properties:
  $\bT_{\overline \Omega}=0$,
  $\iprod{\bT,a_i^{\otimes 3}}=1$,
  and
  \begin{equation}
    \label{eq:tensor-completion-m}
    \set{\norm{x}^2=1}
    \sststile{6}{x}
    \quad
      \iprod{\bT,x^{\otimes 3}}
      \le 1
       -\tfrac 1 {100} \Paren{1-\textstyle \sum_{i=1}^r \iprod{a_i,x}^2}\\
       -\tfrac 1 {100} \Paren{\textstyle \sum_{i\neq j} \iprod{a_i,x}^2\iprod{a_j,x}^2}\,.
  \end{equation}
  These properties imply that $a_1,\ldots,a_r$ are the unique global maximizers of the cubic polynomial $\iprod{\bT,x^{\otimes 3}}$ over the unit sphere.
  (We remark that for matrix completion, the spectral properties of the matrix $M$ imply that the unique global optimizers of the quadratic polynomial $\iprod{M,x^{\otimes 2}}$ are the unit vectors in the span of $a_1,\ldots,a_r$.)

  The proof that this tensor $\bT$ exists follows the same approach as the proof of existence of the matrix $M$ for matrix completion in \cref{matrix-completion-identifiability} and proceeds by an iterative construction \cite{DBLP:journals/jmlr/Recht11,DBLP:journals/tit/Gross11}.
  The main difference is due to the fact that for $M$ we only need to ensure spectral properties, whereas for $\bT$ we need to ensure the existence of (higher-degree) sum-of-squares proofs \cref{eq:tensor-completion-m}.
  We refer to previous literature for details of the proof that such $\bT$ exists with high probability over the choice of $\Omega$ \cite{DBLP:conf/colt/PotechinS17}.

  Similar to the proof for matrix completion, we have by the properties of $\bT$ that $\iprod{\bT,\bX}=r$ and $\cA\sststile{}{}\Set{\iprod{\bT,\sum_{i=1}^r b_i^{\otimes 3}}=\iprod{\bT,\bX}=r}$.
  By \cref{eq:tensor-completion-m} and linearity,
  \begin{displaymath}
    \cA\sststile{}{}
    \iprod{\bT,\textstyle \sum_{i=1}^r b_i^{\otimes 3}} \le r
       - \tfrac 1 {100} \sum_{i=1}^r (1-\textstyle \sum_{j=1}^r \iprod{a_j,b_i}^2)
       - \tfrac 1 {100} \sum_{i=1}^r \sum_{j\neq k} \iprod{a_j,b_i}^2\iprod{a_k,b_i}^2\,.
  \end{displaymath}
    Because $\cA$ includes the equations $\|b_1\|^2 = \cdots = \|b_r\|^2 = 1$ and because the final term is a sum of squares, we conclude that $\cA\sststile{}{} \sum_{j=1}^r \iprod{a_j,b_i}^2=1$ for all $i\in[r]$ and  $\cA\sststile{}{} \iprod{a_j,b_i}^2\iprod{a_k,b_i}^2=0$ for all $i,j,k\in[r]$ with $j\neq k$.
   We also have the following claim:
    \begin{claim}
	    \label[claim]{claim:eq-cube}
	    When $\{a_i\}_{i\in[r]}$ are orthogonal and $\cA \sststile{}{} \left\{\textstyle{\sum_{j \in [r]}} \iprod{a_j,b_i}^2 = 1\right\}_{i \in [r]}$ and $\cA\sststile{}{} \set{\iprod{a_{j_1},b_i}^2\iprod{a_{j_2},b_i}^2 = 0 }_{j_1 \neq j_2 \in [r]}$, then
  \begin{math}
    \cA \sststile{}{} \norm{b_i^{\otimes 3}-\sum_{j=1}^r \iprod{a_j,b_i}^3 a_j^{\otimes 3}}^2 =0\,.
  \end{math}
    \end{claim}
    We give the (easy) proof of \cref{claim:eq-cube} in \cref{app:alg-proofs}.
  Thus, from the orthonormality of the $a_i$,
  \begin{displaymath}
    \cA \sststile {}{} r = \Iprod{\bT,\sum_{i=1}^r b_i^{\otimes 3}}
    = \sum_{i,j} \iprod{a_j,b_i}^3\iprod{\bT, a_j^{\otimes 3}}
    = \sum_{i,j}\iprod{a_j,b_i}^3
    = \Iprod{\bX,\sum_{i=1}^r b_i^{\otimes 3}}\,.
  \end{displaymath}
  \Dnote{}
  \Tnote{}
  Together with the facts $\normf{\bX}^2=r$ and $\cA \sststile{}{} \normf{\sum_{j=1}^r b_i^{\otimes 3}}^2=r$, we obtain $\cA \sststile{}{} \normf{\bX-\sum_{j=1}^r b_i^{\otimes 3}}^2 =0$ as desired.
\end{proof}

\subsection{Overcomplete tensor decomposition}
\label{sec:overc-tens-decomp}

Tensor decomposition refers to the following general class of estimation problems:
Given (a noisy version of) a $k$-tensor of the form $\sum_{i=1}^r a_i^{\otimes k}$, the goal is to (approximately) recover one, most, or all of the component vectors $a_1,\ldots,a_r\in \R^n$.
It turns out that under mild conditions on the components $a_1,\ldots,a_r$, the noise, and the tensor order $k$, this estimation task is possible information theoretically.
For example, generic components $a_1,\ldots,a_r\in \R^n$ with $r\le \Omega(n^2)$ are identified by their 3-tensor $\sum_{i=1}^r a_i^{\otimes 3}$ \cite{DBLP:journals/siammax/ChiantiniO12} (up to a permutation of the components).
Our concern will be what conditions on the components, the noise, and the tensor order allow us to efficiently recover the components.

Besides being significant in its own right, tensor decomposition is a surprisingly versatile and useful primitive to solve other estimation problems.
Concrete examples of problems that can be reduced to tensor decomposition are latent Dirichlet allocation models, mixtures of Gaussians, independent component analysis, noisy-or Bayes nets, and phylogenetic tree reconstruction \cite{DBLP:journals/tsp/LathauwerCC07,DBLP:conf/stoc/MosselR05,DBLP:conf/nips/AnandkumarFHKL12,DBLP:conf/innovations/HsuK13,DBLP:conf/stoc/BhaskaraCMV14,DBLP:conf/stoc/BarakKS15,DBLP:conf/focs/MaSS16,DBLP:conf/stoc/AroraGMR17}.
Through these reductions, better algorithms for tensor decomposition can lead to better algorithms for a large number of other estimation problems.

Toward better understanding the capabilities of efficient algorithms for tensor decomposition, we focus in this section on the following more concrete version of the problem.

\begin{problem}[Tensor decomposition, single component recovery, constant error]
  \label[problem]{prob:tdecomp}
  Given an order-$k$ tensor $\sum_{i=1}^r a_i^{\otimes k}$ with component vectors $a_1,\ldots,a_r\in \R^n$, find a vector $u\in \R^n$ that is close\footnotemark{} to one of the component vectors in the sense that $\max_{i\in [r]}\tfrac1 {\norm{a_i}\cdot \norm{u}}\abs{\iprod{a_i,u}}\ge 0.9$.
\end{problem}

\footnotetext{This notion of closeness ignores the sign of the components.
  If the tensor order is odd, the sign can often be recovered as part of some postprocessing.
  If the tensor order is even, the sign of the components is not identified.
}

Algorithms for \cref{prob:tdecomp} can often be used to solve a priori more difficult versions of the tensor decomposition that ask to recover most or all of the components or that require the error to be arbitrarily small.

A classical spectral algorithm attributed to Jennrich \cite{harshman1970foundations,MR1238921-Leurgans93} can solve \cref{prob:tdecomp} for up to $r\le n$ generic components if the tensor order is at least $3$.
(Concretely, the algorithm works for 3-tensors with linearly independent components.)
Essentially the same algorithm works up to $\Omega(n^2)$ generic\footnote{Here, the vectors $a_1^{\otimes 2},\ldots,a_r^{\otimes 2}$ are assumed to be linearly independent.} components if the tensor order is at least $5$.
A more sophisticated algorithm \cite{DBLP:journals/tsp/LathauwerCC07} solves \cref{prob:tdecomp} for up to $\Omega(n^2)$ generic\footnote{Concretely, the vectors $\set{a_i^{\otimes 2}\otimes a_j^{\otimes 2} \mid i\neq j}\cup \set{(a_i\otimes a_j)^{\otimes 2}\mid i\neq j}$ are assumed to be linearly independent.} components if the tensor order is at least $4$.
However, these algorithms and their analyses break down if the tensor order is only 3 and the number of components is $n^{1+\Omega(1)}$, even if the components are random vectors.

In this and the subsequent section, we will discuss a polynomial-time algorithm based on sum-of-squares that goes beyond these limitations of previous approaches.

\begin{theorem}[\cite{DBLP:conf/focs/MaSS16} building on \cite{DBLP:conf/stoc/BarakKS15,DBLP:conf/approx/GeM15,DBLP:conf/stoc/HopkinsSSS16}]
  \label{overcomplete-tdecomp-sos}
  There exists a polynomial-time algorithm to solve \cref{prob:tdecomp} for tensor order 3 and $\tilde \Omega(n^{1.5})$ components drawn uniformly at random from the unit sphere.
\end{theorem}

The strategy for this algorithm consists of two steps:

\begin{enumerate}
\item use sum-of-squares in order to lift the given order-3 tensor to a noisy version of the order-6 tensor with the same components,
\item apply Jennrich's classical algorithm to decompose this order-6 tensor.
\end{enumerate}

While \cref{prob:tdecomp} falls outside of the scope of \cref{deterministic-meta-theorem} (Meta-theorem for efficient estimation) because the components are only identified up to permutation, the problem of lifting a 3-tensor to a 6-tensor with the same components is captured by \cref{deterministic-meta-theorem}.
Concretely, we can formalize this lifting problem as the following set of parameter--measurement pairs,
\begin{displaymath}
  \cP_{3,6;r} = \Set{(\bX,\bY) \Mid
    \bX=\sum_{i=1}^r a_i^{\otimes 6},~
    \bY=\sum_{i=1}^r a_i^{\otimes 3},~
    a_1,\ldots,a_r\in \R^n}\subseteq \R^{n^6}\times \R^{n^3}\,.
\end{displaymath}
In \cref{sec:tens-decomp-lift}, we give the kind of sum-of-squares proofs that \cref{deterministic-meta-theorem} requires in order to obtain an efficient algorithm to solve the above estimation problem of lifting 3-tensors to 6-tensors with the same components.

The following theorem gives an analysis of Jennrich's algorithm that we can use to implement the second step of the above strategy for \cref{overcomplete-tdecomp-sos}.

\begin{theorem}[Robust Jennrich's algorithm \cite{DBLP:conf/focs/MaSS16,DBLP:conf/colt/SchrammS17}]\torestate{
  \label{robust-jennrich}
  There exists $\e>0$ and a randomized polynomial-time algorithm that given a 3-tensor $\bT\in(\R^n)^{\otimes 3}$ outputs a unit vector $u\in \R^n$ with the following guarantees:
    Let $a_1,\ldots,a_r\in \R^n$ be unit vectors with orthogonality defect $\norm{\Id_r -\transpose{ A} A }\le \e$, where $A\in\R^{n\times r}$ is the matrix with columns $a_1,\ldots,a_r$.
  Suppose $\normf{T-\sum_i a_i^{\otimes 3}}^2\le \e \cdot r$ and that $\max\set{\norm{T}_{\set{1,3}\set{2}},\norm{T}_{\set{1}\set{2,3}}}\le 10$.
    Then, with at least inverse polynomial probability, $\max_{i\in [r]}\iprod{a_i,u}\ge 0.9$.}
\end{theorem}

We will apply \cref{robust-jennrich} to the noisy copy of the $6$-tensor $\bX$ returned by the SoS algorithm, viewing it as a $3$-tensor in the lifted/squared components $a_1 \otimes a_1,\ldots,a_r \otimes a_r \in \R^{n^2}$; these lifted components are in $n^2$ dimensions, and may be linearly independent and close to orthogonal for $r \gg n$.\footnote{To ensure that the lifted/squared components are close to orthogonal, we must stipulate conditions for $\bT$.}
To ensure that our approximation to $\bX$ meets the conditions of the theorem, we can add constraints to the SoS SDP to bound the spectral norm of rectangular reshapings of $\bX$; see \cite{DBLP:conf/focs/MaSS16} for details.
\begin{proof}[Proof sketch]
  We apply the following version of Jennrich's algorithm to $\bT$:
    Choose a Gaussian vector $g \sim \cN(0,\Id)$ and compute the $d\times d$ matrix $T(g)$ given by the random contraction,
    \[
      T(g) = \sum_{j \in [n]}g_j \cdot T_i,
    \]
    where $T_i$ is the $n \times n$ matrix resulting from the restriction of $\bT$ to coordinate $i$ in the third mode.
    Then, output the top eigenvector of $T(g)$.

  To analyze this algorithm, write $\bT$ as a sum of signal and noise terms $\bT=S + E$, where $S=\sum_{i=1}^r a_i^{\otimes 3}$ and $\normf{E}\le \e\cdot r$.
    Notice that when $E = 0$,
    \[
	T(g) = \sum_{i \in [r]} \iprod{g,a_i}\cdot a_i a_i^\top,
    \]
    and with probability $1$ the values $\iprod{a_i,g}$ are distinct.
    So when $\|\Id_r - A^\top A\| = 0$, the eigenvectors of $T(g)$ are exactly the $a_i$.
    To establish the theorem, it remains to show that when $\|E\|_F^2 \le \e r$ and when the orthogonality defect is at most $\e$, the top eigenvector is still close to some $a_i$ with reasonable probability.
    Though the full proof is not complicated, we defer it to \cref{app:alg-proofs}.
\end{proof}

\Dnote{}

\subsection{Tensor decomposition: lifting to higher order}
\label{sec:tens-decomp-lift}

In this section, we give low-degree sum-of-squares proofs of identifiability for the different version of the estimation problem of lifting 3-tensors to 6-tensors with the same components.
These sum-of-squares proofs are a key ingredient of the algorithms for overcomplete tensor decomposition discussed in \cref{sec:overc-tens-decomp}.

We first consider the problem of lifting 3-tensors with orthonormal components.
By itself, this lifting theorem cannot be used for overcomplete tensor decomposition.
However it turns out that this special case best illustrates the basic strategy for lifting tensors to higher-order tensors with the same components.

\paragraph{Orthonormal components.}

The following lemma shows that for orthonormal components, the 3-tensor identifies the 6-tensor with the same set of components and that this fact has a low-degree sum-of-squares proof.

\begin{lemma}
  \label[lemma]{orthogonal-tensor-identifiability-proof}
  Let $a_1,\ldots,a_r\in \R^n$ be orthonormal.
  Let $\cA = \set{\sum_{i=1}^r a_i^{\otimes 3}=\sum_{i=1}^r b_i^{\otimes 3}, \transpose B \cdot B=\Id}$, where $B$ is an $n$-by-$r$ matrix of variables and $b_1,\ldots,b_r$ are the columns of $B$.
  Then,
  \begin{displaymath}
    \cA \sststile{12}{B} \Set{\normf{\sum_{i=1}^r a_i^{\otimes 6} - \sum_{i=1}^r b_i^{\otimes 6}}^2=0}\,.
  \end{displaymath}
\end{lemma}

\begin{proof}
  By orthonormality, $\normf{\sum_{i=1}^r a_i^{\otimes 6} }^2 =\normf{\sum_{i=1}^r a_i^{\otimes 3} }^2 =r$ and from the constraint $B^\top B = \Id$, $\cA \sststile{}{B} \normf{\sum_{i=1}^r b_i^{\otimes 6} }^2=\normf{\sum_{i=1}^r b_i^{\otimes 3} }^2 =r$.
    Thus, by the equality $\sum_i a_i^{\otimes 3} = \sum_j b_j^{\otimes 3}$, we have  $\cA \sststile{}{} \sum_{i,j}\iprod{a_i,b_j}^3 = r$.
    It suffices to show $\cA \sststile{}{} \sum_{i,j}\iprod{a_i,b_j}^6 \ge r$.

  Using $\sum_{i=1}^r \dyad {a_i}\preceq \Id$, a sum-of-squares version of Cauchy--Schwarz, and the fact that $\cA$ contains the constraints $\norm{b_1}^2=\cdots=\norm{b_r}^2=1$,
  \begin{displaymath}
    \cA \sststile{}{} \quad
    r=\sum_{i,j} \iprod{a_i,b_j}^3
    \le \tfrac 12 \sum_{i,j} \iprod{a_i,b_j}^2 + \tfrac 12 \sum_{i,j} \iprod{a_i,b_j}^4
    \le \tfrac 12 r + \tfrac 12 \sum_{i,j} \iprod{a_i,b_j}^4\,.
  \end{displaymath}
  We conclude that $\cA \sststile{}{} \sum_{i,j} \iprod{a_i,b_j}^4=r$.
  Applying the same reasoning to $\sum_{i,j} \iprod{a_i,b_j}^4$ instead of $\sum_{i,j} \iprod{a_i,b_j}^3$ yields $\cA \sststile{}{} \sum_{i,j} \iprod{a_i,b_j}^6=r$ as desired.
\end{proof}

\paragraph{Incoherent components.}

The following lemma shows that a 6-tensor is identifiable from a 3-tensor with the same components if the components satisfy a set of simple deterministic conditions .
Furthermore, this fact has a low-degree sum-of-squares proof.
These conditions allow for overcomplete tensors with components $a_1,\ldots,a_r\in \R^n$ such that $r\ge n^{1+\Omega(1)}$.
In fact, together with the techniques in \cref{sec:overc-tens-decomp}, the following lemma gives a polynomial-time algorithm to solve \cref{prob:tdecomp} for tensor order 3 and up to $\tilde \Omega(n^{1.25})$ components that are drawn uniformly at random from the unit sphere.

For $\sigma \ge 1$ and $\rho>0$, we say that unit vectors $a_1,\ldots,a_r\in \R^n$ are $(\sigma,\rho)$-incoherent if $\sum_{i=1}^r \dyad {a_i}\preceq \sigma \cdot\Id$ and $\abs{\iprod{a_i,a_j}}\le \rho$ for all $i\neq j$.
Random unit vectors satisfy this property for $\sigma \le \tilde O(r/n)$ and $\rho\le \tilde O(1/\sqrt n)$.

Let $B$ be an $n$-by-$r$ matrix of variables and let $b_1,\ldots,b_r$ be the columns of $B$.
Consider the following system of polynomial constraints
\begin{equation}
  \label{eq:normf-constraints}
    \cB_{\e} = \Set{\norm{b_i}^2=1\,\, \forall i \in [r],~
    \normf{\textstyle \sum_{i=1}^r b_i^{\otimes 3}}^2 \ge (1-\e) \cdot r,~
    \normf{\textstyle \sum_{i=1}^r b_i^{\otimes 6}}^2 \le (1+\e) \cdot r}\,.
\end{equation}
We observe that $(\sigma,\rho)$-incoherent unit vectors satisfy $\cB_\e$ for $\e=\rho \sigma$.
In particular, if $a_1,\ldots,a_r$ are $(\sigma,\rho)$-incoherent unit vectors, then $\normf{\sum_{i=1}^r a_i^{\otimes 3}}^2 = r + \sum_{i\neq j}\iprod{a_i,a_j}^3\ge (1-\rho \sigma)\cdot r$.
For a similar reason, $\normf{\sum_{i=1}^r a_i^{\otimes 6}}^2\le (1+\rho^4 \sigma)\cdot r\le (1+\rho \sigma)\cdot r$.

\begin{lemma}
  \label[lemma]{incoherent-tensor-identifiability-proof}
  Let $a_1,\ldots,a_r\in \R^n$ be $(\sigma,\rho)$-incoherent unit vectors.
  Let $B$ be an $n$-by-$r$ matrix of variables, $b_1,\ldots,b_r$ the columns of $B$, and $\cA$ the following system of polynomial constraints,
  \begin{displaymath}
    \cA = \cB_{\rho \sigma} \bigcup \Set{\sum_{i=1}^r a_i^{\otimes 3}=\sum_{i=1}^r b_i^{\otimes 3}}\,.
  \end{displaymath}
  Then,
  \begin{displaymath}
    \cA \sststile{12}{B} \Set{\normf{\sum_{i=1}^r a_i^{\otimes 6} - \sum_{i=1}^r b_i^{\otimes 6}}^2\le O(\rho \sigma^2) \cdot \normf{\sum_{i=1}^r a_i^{\otimes 6} + \sum_{i=1}^r b_i^{\otimes 6}}^2}\,.
  \end{displaymath}
\end{lemma}

The proof follows the same strategy as our proof of \cref{orthogonal-tensor-identifiability-proof}.
We aim to lower bound first $\sum_{i,j}\iprod{a_i,b_j}^4$ and then $\sum_{i,j}\iprod{a_i,b_j}^6$.

\begin{proof}
  Since $\normf{\sum_{i=1}^r a_i^{\otimes 3}}^2\ge (1-\rho \sigma)\cdot r$ and  $\cA \sststile{}{} \normf{\sum_{i=1}^r b_i^{\otimes 3}}^2 \ge (1-\rho \sigma)\cdot r$, it holds that $\cA \sststile{}{} \sum_{i,j}\iprod{a_i,b_j}^3 \ge (1-\rho \sigma)\cdot r$.
  At the same time, since $\normf{\sum_{i=1}^r a_i^{\otimes 6}}^2\le (1+\rho \sigma)\cdot r$ and  $\cA \sststile{}{} \normf{\sum_{i=1}^r b_i^{\otimes 6}}^2\le (1+\rho \sigma)\cdot r$, it suffices to show $\cA \sststile{}{} \sum_{i,j}\iprod{a_i,b_j}^6 \ge (1-10\rho \sigma^2)\cdot r$.
  Indeed,
  \footnotetext{
    A formal reason for these bounds is that the assignment $b_i=a_i$ satisfies the constraints $\cA$.
  }
  \begin{displaymath}
    \cA \sststile{}{}\quad
    \begin{aligned}[t]
      (1-\rho \sigma)\cdot r
      & \le \textstyle \sum_{i,j}\iprod{a_i,b_j}^3 \\
      & = \textstyle \sum_{j}\Iprod{b_j, \sum_{i} \iprod{a_i,b_j}^2 a_i} \\
      & \le \textstyle \sum_{j} \tfrac 12 \norm{b_j}^2 + \tfrac 12\Norm{\sum_{i} \iprod{a_i,b_j}^2 a_i}^2\\
      & = \textstyle \tfrac 12 r
      + \tfrac 12\sum_{i,j} \iprod{a_i,b_j}^4 + \tfrac 12 \sum_{j} \sum_{i\neq i'} \iprod{a_i,a_{i'}} \iprod{a_i,b_j}^2\iprod{a_{i'},b_j}^2\\
      & \le \tfrac 12 r + \tfrac 12\sum_{i,j} \iprod{a_i,b_j}^4 + \tfrac 12 \rho \sigma^2 r,
      \,.
    \end{aligned}
  \end{displaymath}
    where to obtain the final line we have used that $|\iprod{a_i,a_{i'}}| \le \rho$ and $\sum_{i=1}^r \iprod{a_i,b_j}^2 = b_j^\top \left(\sum_{i=1}^r a_i a_i^\top\right)b_j \le \sigma$ by the assumption that $\sum_i a_i a_i^\top \preceq \sigma \Id$.
  It follows that $\cA \sststile{}{} \sum_{i,j} \iprod{a_i,b_j}^4 \ge (1-\rho \sigma^2 - 2\rho \sigma)\cdot r$.
  By applying the above reasoning to $\sum_{i,j} \iprod{a_i,b_j}^4$ instead of $\sum_{i,j} \iprod{a_i,b_j}^3$, we obtain $\cA \sststile{}{} \sum_{i,j} \iprod{a_i,b_j}^6 \ge (1 - 3\rho \sigma^2 - 4\rho \sigma)\cdot r\ge (1-10\rho\sigma^2)\cdot r$ as desired.
  Concretely,
  \begin{displaymath}
    \cA \sststile{}{}\quad
    \begin{aligned}[t]
      (1-\rho \sigma^2 - 2\rho\sigma)\cdot r
      & \le \textstyle \sum_{i,j}\iprod{a_i,b_j}^4 \\
      & = \textstyle \sum_{j}\Iprod{b_j, \sum_{i} \iprod{a_i,b_j}^3 a_i} \\
      & \le \textstyle \tfrac 12 r
      + \tfrac 12\sum_{i,j} \iprod{a_i,b_j}^6 + \tfrac 12 \sum_{j} \sum_{i\neq i'} \iprod{a_i,a_{i'}} \iprod{a_i,b_j}^3\iprod{a_{i'},b_j}^3\\
      & \le \textstyle \tfrac 12 r
      + \tfrac 12\sum_{i,j} \iprod{a_i,b_j}^6 + \tfrac 12 \rho \sum_{j} \sum_{i\neq i'} \iprod{a_i,b_j}^2\iprod{a_{i'},b_j}^2\\
      & \le \tfrac 12 (1+\rho \sigma^2)\cdot r + \tfrac 12\sum_{i,j} \iprod{a_i,b_j}^6
      \,.\qedhere
    \end{aligned}
  \end{displaymath}
\end{proof}

\paragraph{Random components.}

Let $a_1,\ldots,a_r\in \R^n$ be uniformly random unit vectors with $r\le n^{O(1)}$.
Let $B$ be an $n$-by-$r$ matrix of variables and let $b_1,\ldots,b_r$ be the columns of $B$.
With high probability, the vectors $a_1,\ldots,a_r$ satisfy $\cB_\e$ for $\e \le \tilde O(r/n^{1.5})$, as defined in \cref{eq:normf-constraints}.
Concretely, with high probability, every pair $(i,j)\in [r]^2$ with $i\neq j$ satisfies $\iprod{a_i,a_j}^2\le \tilde O(1/ n)$.
Thus, $\normf{\sum_{i=1}^r b_i^{\otimes 3}}^2 = r + \sum_{i\neq j} \iprod{b_i,b_j}^3\ge (1+ \tilde O(r/n^{1.5}))\cdot r$ and $\normf{\sum_{i=1}^r b_i^{\otimes 6}}^2 \le (1+\tilde O(r/n^3))\cdot r$.

\begin{lemma}[implicit in \cite{DBLP:conf/approx/GeM15}]
  \label[lemma]{random-tensor-identifiability-proof}
  Let $\e>0$ and $a_1,\ldots,a_r\in \R^n$ be random unit vectors with $r\le \e\cdot\tilde \Omega(n^{1.5})$.
  Let $B$ be an $n$-by-$r$ matrix of variables, $b_1,\ldots,b_r$ the columns of $B$, and $\cA$ the following system of polynomial constraints,
  \begin{displaymath}
    \cA = \cB_{\e} \bigcup \Set{\sum_{i=1}^r a_i^{\otimes 3}=\sum_{i=1}^r b_i^{\otimes 3}}\,.
  \end{displaymath}
  Then,
  \begin{displaymath}
    \cA \sststile{12}{B} \Set{\normf{\sum_{i=1}^r a_i^{\otimes 6} - \sum_{i=1}^r b_i^{\otimes 6}}^2 \le O(\e) \cdot \normf{\sum_{i=1}^r a_i^{\otimes 6} + \sum_{i=1}^r b_i^{\otimes 6}}^2}\,.
  \end{displaymath}
\end{lemma}

\begin{proof}
  With high probability over the choice of $a_1,\ldots,a_n$, we have $\normf{\sum_{i=1}^r a_i^{\otimes 3}}\ge( 1-\e)\cdot r$ and $\normf{\sum_{i=1}^r a_i^{\otimes 6}}\ge( 1+\e)\cdot r$.
  Therefore, it holds $\cA \sststile{}{} \sum_{i,j}\iprod{a_i,b_j}^3\ge 1-\e$ and it suffices to show $\cA \sststile{}{} \sum_{i,j} \iprod{a_i,b_j}^6\ge 1-10\e$.

  The work \cite{DBLP:conf/approx/GeM15} shows that, with high probability over the choice of $a_1,\ldots,a_n$,
  \begin{displaymath}
    \Set{\norm{x}^2=1}  \sststile{}{}
    \quad\Set{
      \begin{aligned}
        \sum_{i\neq j} \iprod{a_i,a_j} \iprod{a_i,x}^2\iprod{a_j,x}^2 \le \e\\
        \sum_{i\neq j} \iprod{a_i,a_j} \iprod{a_i,x}^3\iprod{a_j,x}^3 \le \e\\
      \end{aligned}
    }\,.
  \end{displaymath}
  Under these conditions, the same reasoning as in the proof of \cref{incoherent-tensor-identifiability-proof} allows us to conclude $\cA \sststile{}{} \sum_{i,j} \iprod{a_i,b_j}^4 \ge (1-3\e)\cdot r$ and  $\cA \sststile{}{} \sum_{i,j} \iprod{a_i,b_j}^6 \ge (1-7\e)\cdot r$.
\end{proof}

\Dnote{}

\subsection{Clustering}

We consider the following clustering problem:
given a set of points $y_1,\ldots,y_n\in\R^d$, the goal is to output a $k$-clustering matrix $X\in \bits^{n\times n}$ of the points such that the points in each cluster are close to each other as possible.
Here, we say that a matrix $X\in \bits^{n\times n}$ is a $k$-clustering if there is a partition $S_1,\ldots,S_k$ of $[n]$ such that $X_{ij}=1$ if and only if there exists $\ell\in[k]$ with $i,j\in S_\ell$.

In this section, we will discuss how SoS allows us to efficiently find clusterings with provable guarantees that are significantly stronger than for previous approaches.
For concreteness, we consider in the following theorem the extensively studied special case that the points are drawn from a mixture of spherical Gaussians such that the means are sufficiently separated \cite{DBLP:conf/focs/Dasgupta99,DBLP:conf/stoc/SanjeevK01,DBLP:journals/jcss/VempalaW04,DBLP:conf/colt/AchlioptasM05,DBLP:conf/stoc/KalaiMV10,DBLP:conf/focs/MoitraV10,DBLP:conf/colt/BelkinS10}.
Another key advantage of the approach we discuss is that it continues to work even if the points are not drawn from a mixture of Gaussians and the clusters only satisfy mild bounds on their empirical moment tensors.

\begin{theorem}[\cite{HopkinsL18,KothariSS18,DiakonikolasKS18}]
  \label{clustering-algorithm}
  There exists an algorithm that given $k \in \N$ with $k\le n$ and vectors $\hat y_1,\ldots,\hat y_n\in \R^d$ outputs a $k$-clustering matrix $X\in\bits^{n\times n}$ in quasi-polynomial time $n+\cramped{(dk)^{(\log k)^{O(1)}}}$ with the following guarantees:
  Let $\hat y_1,\ldots,\hat y_n$ be a sample from the uniform mixture of $k$ spherical Gaussians $\cN(\mu_1,\Id),\ldots,\cN(\mu_k,\Id)$ with mean separation $\min_{i\neq j}\norm{\mu_i-\mu_j}\ge O(\sqrt{\log k})$ and $n\ge \cramped{(dk)^{(\log k)^{O(1)}}}$.
  Let $X^*\in\bits^{n\times n}$ be the $k$-clustering matrix corresponding to the Gaussian components (so that $X^*_{ij}=1$ if $\hat y_i$ and $\hat y_j$ were drawn from the same Gaussian component and $X^*_{ij}=0$ otherwise).
  Then with high probability,
  \begin{displaymath}
    \normf{X - X^*}^2 \le 0.1\cdot\normf{X^*}^2.
  \end{displaymath}
\end{theorem}

We remark that the same techniques also give a sequence of polynomial-time algorithms that approach the logarithmic separation of the algorithm above.
Concretely, for every $\e>0$, there exists an algorithm that works if the mean separation is at least $O_\e(k^{\e})$.

These algorithms for clustering points drawn from mixtures of separated spherical Gaussians constitute a significant improvement over previous algorithms that require separation at least $O(k^{1/4})$ \cite{DBLP:journals/jcss/VempalaW04}.

\paragraph{Sum-of-squares approach to learning mixtures of spherical Gaussians.}
In order to apply \cref{deterministic-meta-theorem}, we view the clustering matrix $X$ corresponding to the Gaussian components as the parameter and a ``typical sample'' $ Y = y_1,\ldots,y_n$ of the mixture as the measurement.
Here, typical means that the empirical moments in each cluster are close to the moments of a spherical Gaussian distribution.
Concretely, we consider the following set of parameter--measurement pairs,
\begin{displaymath}
  \cP_{k,\e,\ell} = \Set{(X,Y) \Mid
    \begin{minipage}{9.5cm}
      $X$ is $k$-clustering matrix w/clusters $S_1,\ldots,S_k\subseteq [n]$\\
	$\forall \kappa \in [k],\\ \hphantom{,,,} \normf{\E_{i\in S_\kappa} (1,y_i-\mu_\kappa)^{\otimes \ell}-\E_{g\sim N(0,\Id)}(1,g)^{\otimes \ell}}\le \e$
    \end{minipage} } \subseteq \bits^{n\times n} \times \R^{d\times n}\,,
\end{displaymath}
where $\mu_\kappa=\E_{i\in S_\kappa} y_i$ is the mean of cluster $S_\kappa\subseteq [n]$, \Tnote{}
and where $(1,v)$ is the vector of dimension $\dim(v) + 1$ with a $1$ in the first coordinate (we extend $y_i$ and $g$ in this way so that the bound includes all moments of order at most $\ell$).

It is straightforward to express $\cP_{k,\e,\ell}$ in terms of a system of polynomial constraints $\cA=\set{p(X,Y,z)=0}$, so that $\cP_{k,\e,\ell}=\set{(X,Y) \mid \exists z.~ p(X,Y,z)=0}$.
\cref{clustering-algorithm} follows from \cref{deterministic-meta-theorem} using the fact that under the conditions of \cref{clustering-algorithm}, the following sum-of-squares proof exists with high probability for $\ell\le (\log k)^{O(1)}$,
\begin{displaymath}
    \cA(\widehat Y) \sststile{\ell}{X,z} \Set{\normf{X - X^*}^2 \le 0.1\cdot\normf{X^*}^2}\,,
\end{displaymath}
where $X^*$ is the ground-truth clustering matrix (corresponding to the Gaussian components), and $\widehat Y = \hat y_1, \ldots,\hat y_n$ is the input samples observed.
\Dnote{}
\Tnote{}

%% file: content/lowerbounds.tex
\newcommand{\opY}{\Pi^{\cY}}
\newcommand{\opX}{\Pi^{\cX}}
\crefname{appendix}{Appendix}{Appendices}

\section{Lower bounds}\label[section]{sec:lowerbounds}
In this section, we will be concerned with showing lower bounds on the minimum degree of sum-of-squares refutations for polynomial systems, especially those arising out of estimation problems.

The turn of the millennium saw several works that rule out degree-$2$ sum-of-squares refutations for a variety of problems, such as \maxcut \cite{MR1900615-Feige02}, \kclique \cite{MR1742351-Feige00}, and \sparsestcut \cite{MR3323774-Khot15},  among others.
These works, rather than explicitly taking place in the context of sum-of-squares proofs, were motivated by the desire to show tightness for specific SDP relaxations.

Around the same time, Grigoriev proved {\em linear} lower bounds on the degree of sum-of-squares refutations for $k$-XOR, $k$-SAT, and knapsack \cite{DBLP:journals/tcs/Grigoriev01,DBLP:journals/cc/Grigoriev01}
(the first two of these bounds were later independently rediscovered by Schoenebeck \cite{DBLP:conf/focs/Schoenebeck08}).
Few other lower bounds against SoS were known.
Most of the subsequent works (e.g. \cite{DBLP:conf/stoc/Tulsiani09,DBLP:conf/soda/BhaskaraCVGZ12}) built on the $k$-SAT lower bounds via reductions; in essence, techniques for proving lower bounds against higher-degree sum-of-squares refutations were ad hoc and few.

In recent years, a series of papers \cite{DBLP:conf/stoc/MekaPW15,DBLP:conf/colt/DeshpandeM15,DBLP:conf/soda/HopkinsKPRS16} introduced higher-degree sum-of-squares lower bounds for \kclique, culminating in the work of Barak \etal \cite{DBLP:conf/focs/BarakHKKMP16}.
Barak \etal\ go beyond proving lower bounds for the \kclique problem specifically, introducing a beautiful and general framework, called {\em pseudocalibration}, for proving SoS lower bounds.
Though their work settles the degree of SoS refutations for \kclique in $\gnhalf$, it brings up intriguing new questions.
In particular, it gives rise to a compelling conjecture, which if proven, would settle the degree needed to refute a broad class of estimation problems, including \densestksubgraph, community detection problems, graph coloring, and more.
We devote this section to describing the technique of pseudocalibration.

Let us begin by recalling some notation.
Let $\cP = \{ p_i(x,y) \geq 0 \}_{i \in [m]}$ be a polynomial system associated with an estimation problem.  The polynomial system is over hidden variables $x \in \R^n$, with coefficients that are functions of the measurement/instance variables $y\in \R^N$.
We will use $\cP_y$ to denote the polynomial system for a fixed $y$.
Let $\cP$  have degree at most $d_x$ in $x$ and degree at most $D_y$ in $y$.
If $\nulld$ denotes the null distribution, then $\cP_y$ is infeasible w.h.p. when $y \sim \nulld$, and we are interested in the minimum degree of sum-of-squares refutation.
\paragraph{Pseudodensities.}
By \cref{claim:functionalNoRefutation}, to rule out degree-$d$ sum-of-squares refutations for $\cP_y$ , it is sufficient to construct the dual pseudoexpectation functional $\pE_y$ with the properties outlined in \cref{sec:sos}.
However, it turns out to be conceptually cleaner to think about constructing related objects called \emph{pseudodensities} rather than \emph{pseudoexpectation functionals}.
Towards defining pseudodensities, we first pick a natural background measure $\xdist$ for $x \in \R^n$,
and we use $\E_{x}$ to denote the expectation over the background measure $\xdist$.
The choice of background measure itself is not too important, but for the example we will consider, it will be convenient to pick $\xdist$ to be uniform distribution over $\bits^n$.
\begin{definition}
A function $\bar \mu : \bits^n \to \R$ is a pseudodensity for a polynomial system $\cP = \{ p_i(x) \geq 0\}_{i \in [m]}$ if $\pE_{\bar \mu} : \R[x]_{\leq d} \to \R$ defined as follows:
\[ \pE_{\bar \mu} [p(x)] \defeq \E_{x} \bar \mu(x) p(x)\]
 is a valid pseudoexpectation operator, namely, it satisfies the constraints outlined in \cref{sec:sos}.
\end{definition}

To show that $\cP_y$ does not admit a degree $d$ SoS refutation for most $y \sim \nulld$, it suffices for us to show that with high probability over $y\sim \nulld$, we can construct a pseudodensity $\bar \mu_y : \bits^n \to \R$.
More precisely, with high probability over the choice of $y \sim \nulld$, the following must hold:
\begin{align}
&(\text{scaling})    &\E_{x} \bar \mu_y(x)     = 1  &\label{eq:scaling} \\
&(\text{PSDness})    & \E_{x} q(x)^2 \bar \mu_y(x)   \ge 0    & \qquad \forall q \in \R[x]_{\leq d/2}\label{eq:psdness} \\
&(\cP\text{ constraints})    & \E_{x} p(x) a^2(x) \cdot \bar \mu_y(x)   \geq  0 & \qquad \forall p \in \cP, a \in \R[x],\deg(a^2 \cdot p) \le d \label{eq:sat}.
\end{align}

\subsection{Pseudocalibration}
{\em Pseudocalibration} is a heuristic for constructing pseudodensities for non-feasible systems in such settings.
It was first introduced in \cite{DBLP:conf/focs/BarakHKKMP16} for the \kclique problem, but the heuristic is quite general and can be seen to yield lower bounds for other problems as well (e.g. \cite{DBLP:journals/tcs/Grigoriev01, DBLP:conf/focs/Schoenebeck08}).

At a high level, pseudocalibration leverages the existence of a structured distribution of estimation problems to construct pseudodensities.
For each $x \in \{0,1\}^n$, let $\cD_x$ be a distribution over $\{\pm 1\}^N$ such that $(x,y)$ are a feasible pair for $\cP$.\footnote{Again, the choice $y \in \{\pm 1\}^N$ is not fundamental, and we make it for for simplicity of presentation.
Also, for calculations it will sometimes be convenient to define $\cD_x$ so that $(x,y)$ is feasible only with high probability over $y\sim\cD_x$; however this does not greatly impact the arguments, and we neglect this detail in our exposition.}
Let $\jfdist$ denote the joint \emph{structured} distribution over feasible pairs $y^* \in \{\pm 1\}^N$ and $x^*$ sampled from $\xdist$, i.e. $\Pr_{\jfdist}\{ (x,y) \} = \xdist(x) \cdot \Pr_{\cD_x}\{ y \}$.
Let us define a joint null distribution $\jndist$ on pairs $(x,y)$ to be
\[
\jndist \defeq \xdist \times \nulld \mper
\]
As we describe pseudocalibration, $\jndist$ will serve as the background measure for us.
Let $\mu_* : \bits^n \times \sbits^N \to \R^+$ denote the density of the joint structured distribution $\jfdist$ with respect to the background measure $\jndist$, namely
\[ \mu_*(x,y) = \frac{\Pr_{\jfdist}(x,y)}{\Pr_{\jndist}(x,y)} = \frac{\Pr_{\feasd} \{ y \} } {\Pr_{\nulld}\{y\}} \cdot \frac{\Pr_{\jfdist}\{x|y\}}{ \sigma(x) }  \]

At first glance, a candidate construction of a pseudodensity $\bar \mu_y$ for $y \sim\nulld$ would be the partially-evaluated relative joint density $\mu_*$ namely
\[
   \bar \mu_y = \mu_*(\cdot,y) \mper
\]

This construction $\bar \mu_y$ already satisfies two of the three conditions  for being a pseudodensity (\cref{eq:psdness} and \cref{eq:sat}).
This is because for any polynomial $p(x,y)$,
\[ \E_x  p(x) \bar \mu_y(x)  ~=~  \frac{\Pr_{\feasd} \{ y \} } {\Pr_{\nulld}\{y\}} \cdot \E_{x} p(x)   \frac{\Pr_{\jfdist}\{x|y\}}{ \sigma(x) }  ~=~ \frac{\Pr_{\feasd} \{ y \} } {\Pr_{\nulld}\{y\}} \cdot  \E_{x \sim \jfdist{(\cdot | y)}} p(x)   \mper \]
From the above equality, \cref{eq:psdness} follows directly because
\[ \E_x  q(x)^2 \bar \mu_y(x)  ~=~  \frac{\Pr_{\feasd} \{ y \} } {\Pr_{\nulld}\{y\}} \cdot  \E_{x \sim \jfdist{(\cdot | y)}} q^2(x)  \ge 0 \mper \]
Similarly, \cref{eq:sat} is again an immediate consequence of the fact that $\jfdist$ is supported on feasible pairs for $\cP$,
\[ \E_x  p(x)a^2(x) \bar \mu_y(x)  = \frac{\Pr_{\feasd} \{ y \} } {\Pr_{\nulld}\{y\}} \cdot \E_{x \sim \jfdist{(\cdot | y)}} p(x) a^2(x)  \ge 0 \mper\]

However, the scaling constraint \cref{eq:scaling} is far from satisfied because,
\[ \E_x \bar \mu_y(x) = \frac{\Pr_{\feasd} \{ y \} } {\Pr_{\nulld}\{y\}} \cdot \E_{x \sim \jdist(\cdot|y)} 1 = \frac{\Pr_{\feasd} \{ y \} } {\Pr_{\nulld}\{y\}}  \]
is a quantity that may be really large for $y \in \supp(\feasd)$ and $0$ otherwise (recall that $\feasd$ has low entropy compared to $\nulld$).
As a saving grace, the constraint \cref{eq:scaling} is satisfied in expectation over $y$, i.e.,
\[ \E_{y \sim \nulld} \E_x \bar \mu_y(x) = \E_{y \sim \nulld} \E_x \mu^{*}(x,y) = \E_{(x,y) \sim \jndist} \mu^*(x,y) = 1 \mcom\]
since $\mu^*$ is a density.

The relative joint density $\mu_*(x,y)$ faces an inherent limitation in that it is only nonzero on $\supp(\feasd)$, which accounts for a negligible fraction of $y \sim \nulld$.
Intuitively, the constraints of $\cP$ are low-degree polynomials in $x$ and $y$.  Therefore, our goal is to construct a $\bar \mu_y$ that has the same low-degree structure as $\mu_*$, but has a much higher entropy: that is, its mass is not concentrated on a small fraction of instances.

A natural way to achieve this is to simply project the joint density $\mu_*$ in to the space of low-degree polynomials.
Formally, let $L_2(\jndist)$ denote the vector space of functions over $\R^N \times \R^n$ equipped with the inner product $\iprod{f,g}_{\jndist} = \E_{(x,y) \sim \jndist} f(x,y) g(x,y)$.
For $d,D \in \N$, let $V_{d,D} \subseteq L_2(\jndist)$ denote the following vector space
\[V_{d,D} = \Span\{ q(x,y) \in \R[x,y]| \deg_x(q) \leq d , \deg_y(q) \leq D \}\]
If $\Pi_{d,D}$ denotes the projection on to $V_{d,D}$, then the pseudo-calibration recipe suggests the use of the following candidate pseudodensity:
\begin{definition}
    For $D \in \mathbb{N}$, the {\em $D$-pseudocalibrated function} $\bar\mu(x,y)$ is defined as
\begin{equation} \label{eq:pseudocalib}
\bar \mu_y(x) = \Pi_{d,D} \circ \mu_*( x , y)
\end{equation}
    where $d$ is the target degree for the pseudodistribution.
\end{definition}

Consider a constraint  in the polynomial system $\{ p(x,y) \geq 0 \} \in \cP$.
As long as $\deg_x(p) \leq d$ and $\deg_y(p) \leq D$, the pseudodensity $\bar \mu_y$ satisfies the constraint in expectation over $y$.
This is immediate from the following calculation,
\begin{align*}
\qquad\qquad\qquad\qquad \E \bar \mu_y(x) p(x,y) & =  \E_{(x,y) \sim \jndist} (\Pi_{d, D} \circ \mu_*(x,y)) p(x,y) \\
    & = \E_{(x,y) \sim \jndist}  \mu_*(x,y) p(x,y) \qquad\qquad (\text{because } p \in V_{d,D})  \\
  &  = \E_{(x,y) \sim \jfdist} p(x,y) \geq 0 \mper
\end{align*}

We additionally require that the constraints of this form are satisfied for each $y \sim \nulld$, and not just in expectation.
Often, this follows using fairly straightforward arguments.
In fact, for equality constraints constraints of the form $\{p(x,y) = 0\}$, one can show that the pseudocalibrated construction satisfies these  constraints with high probability under very mild conditions on the joint distribution $\jfdist$.
Specifically, the following theorem holds.
\begin{theorem}\torestate{ \label{thm:pseudocalib-constraints}
    Suppose $\{ p(x,y) = 0 \} \in \cP$ is always satisfied for $(x,y) \sim \jfdist$ and let $B \seteq \max_{(x,y) \in \jndist} |p(x,y)|$ and let $D_y \seteq \deg_y(p)$ and $d_x \seteq \deg_x(p)$.
    If $d \ge d_x$ and $\bar \mu_y$ is the $D$-pseudocalibrated function defined in \cref{eq:pseudocalib} then
\[ \Pr_{y \sim \nulld} [ \abs{\E_{x} p(x,y) \bar \mu_y(x)}  \geq \epsilon ]  \leq   \frac{B^2}{\epsilon^2}\cdot \Norm{\Pi_{d, D + 2D_y} \circ \mu_* - \Pi_{d,D-1} \circ \mu_*  }_{2,\jndist}^2 \]
    where $\Pi_{d, D}$ for $d, D \in \N$ denotes the projection on to $V_{d,D}$, the span of polynomials of degree at most $D$ in $y$  and degree $d$ in $x$.
    }
\end{theorem}
The theorem suggests that if the projection of the structured density $\mu_*$ decays with increasing degree in $y$, then for $D$ chosen large enough, the $D$-pseudocalibrated function $\bar \mu_y$ satisfies the same equality constraints as those satisfied by $\mu_*$, with high probability.
This decay in the Fourier spectrum of the structured density is a common feature in all known applications of pseudocalibration.
We defer the proof of the \cref{thm:pseudocalib-constraints} to \cref{app:proof-of}.

\paragraph{Verifying non-negativity of squares.}
The chief obstacle in establishing $\bar\mu(\cdot,y)$ as a valid pseudodensity is in proving that it satisfies the constraint $\E_{x} p(x,y)^2 \bar \mu(x,y) \ge 0$, for every polynomial $p$ of degree at most $\tfrac{d}{2}$ in $x$.
As we will see in \cref{claim:squaresimpliespsd}, this condition is equivalent to establishing the positive-semidefiniteness (PSDness) of the matrix
\begin{equation}
    M_d(y) \defeq \E_x \left[\left(x^{\le d/2}\right)\left(x^{\le d/2}\right)^\top \cdot \bar\mu(x,y)\right],\label{eq:mom-mat}
\end{equation}
where $x^{\le d/2}$ is the $O(n^{d/2})\times 1$ vector whose entries contain all monomials of degree at most $\frac{d}{2}$ in $x$.
\begin{claim} \label[claim]{claim:squaresimpliespsd}
$\E_{x} q(x,y)^2 \bar \mu(x,y) \ge 0$ for all polynomials $q(x,y)$ of degree at most $d/2$ in $x$ if and only if the matrix $M_d(y)$ is positive semidefinite.
\end{claim}
\begin{proof}
    The first direction is given by expressing $q(x,y)$ with its vector of coefficients of monomials of $x$, $\hat q(y)$, so that $\iprod{\hat q(y), x^{\le d/2}} = q(x,y)$.
    Then
    \[
	\E_{x} q(x,y)^2 \bar \mu(x,y) = \E_{x} [ \hat q(y)^\top (x^{\le d/2})(x^{\le d/2})^\top \hat q(y)\cdot \mu(x,y)] = \hat q(y)^\top M_d(y) \hat q(y) \ge 0,
    \]
by the positive-semidefiniteness of $M(y)$.

    To prove the contrapositive, we note that if $M_d(y)$ is not positive-semidefinite, then there is some negative eigenvector $v(y)$ so that $v(y)^\top M_d(y) v(y) < 0$.
    Taking $q(x,y) = \iprod{v(y),x^{\le d/2}}$, we have our conclusion.
\end{proof}

Each entry of $M_d(y)$ is a degree-$D$ polynomial in $y \sim \nulld$.
Since the entries of $M_d(y)$ are not independent, and because $M_d(y)$ cannot be decomposed easily into a sum of independent random matrices, standard black-box matrix concentration arguments such as matrix Chernoff bounds and Wigner-type laws do not go far towards characterizing the spectrum of $M_d(y)$.
For this reason proving PSDness for $M_d(y)$ is a delicate process.
Though the known lower bounds for planted clique, random SAT refutation, and other problems all use the same construction for $\bar \mu$, the current proofs of PSDness are very tailored to the specific choice of $\nulld$, and in some cases they are quite technical.
We will expand further in \cref{sec:matrix-poly}.

\subsubsection*{Pseudocalibration: a partial answer, and many questions}
While \cref{thm:pseudocalib-constraints} establishes some desirable properties for the pseudocalibrated function $\bar \mu$, we are left with many unanswered questions.
Ideally, we would be able to identify simple, general sufficient conditions on the structured distribution $\feasd$ and on $d$ the degree in $x$ and $D$ the degree in $y$, for which $\bar\mu$ yields a valid pseudodensity.
The following conjecture stipulates one such choice of conditions:
\begin{conjecture}\label[conjecture]{conj:low-deg}
    Suppose that $\cP$ contains no polynomial of degree more than $k$ in $y$.
    Let $D = O(k d \log n)$ and $D = \Omega(kd)$.
    Then the $D$-pseudocalibrated function $\bar\mu(\cdot, y)$ is a valid degree-$d$ pseudodistribution which satisfies $\cP$ with high probability over $y \sim \nulld$ if and only if there is no polynomial $q(y)$ of degree at most $D$ in $y$ such that $\E_{y \sim\nulld}[q(y)] = 0$ and
    \[
	n^{\omega(d)} \cdot \sqrt{\E_{y \sim \nulld}[q(y)^2]} < \E_{y \sim \feasd} [q(y)].
    \]
\end{conjecture}
The upper and lower bounds on $D$ stated in \cref{conj:low-deg} may not be precise; what is important is that $D$ not be too much larger than $O(kd)$.
In support of this conjecture, we list several refutation problems for which the conjecture has been proven: \kclique \cite{DBLP:conf/focs/BarakHKKMP16}, \tensorpca \cite{DBLP:conf/focs/HopkinsKPRSS17}, and random $k$-SAT and $k$-XOR \cite{DBLP:journals/tcs/Grigoriev01,DBLP:conf/focs/Schoenebeck08}.
However, in each of these cases, the proofs have been somewhat ad hoc, and do not generalize well to other problems of interest, such as densest-$k$-subgraph, community detection, and graph coloring.

Resolving this conjecture, which will likely involve discovering the ``book'' proof of the above results, is an open problem which we find especially compelling.

\paragraph{Variations.}
The incompleteness of our understanding of the pseudocalibration technique begs the question, is there a different choice of function $\mu'(x,y)$ such that $\mu'(\cdot, y)$ is a valid pseudodensity satisfying $\cP$  with high probability over $y\sim\nulld$?

Indeed, already among the known constructions there is some variation in the implementation of the low-degree projection: the truncation threshold is not always a sharp degree $D$, and is sometimes done in a gradual fashion to ease the proofs (see e.g. \cite{DBLP:conf/focs/BarakHKKMP16}).
It is a necessary condition that $\mu'$ and $\mu_*$ agree at least on the moments of $y$ which span the constraints of $\cP$ (otherwise $\mu'$ cannot satisfy $\cP$ in expectation).
However, there are alternative ways to ensure this, while also choosing $\mu'$ to have higher entropy than $\mu_*$.

In \cite{DBLP:conf/focs/HopkinsKPRSS17}, the authors give a different construction, in which rather than setting $\mu' = \Pi_{d,D}  \mu_*$, they choose the function $\mu'$ which minimizes the Frobenius norm under the constraint that $\E_x x^{\le d/2}(x^{\le d/2})^\top \mu'(x,y)$ is positive semidefinite for every $y \in \supp(\nulld)$, and that $\Pi_{d,D} \mu'(x,y)  = \Pi_{d,D} \mu_*(x,y) $. Though in \cite{DBLP:conf/focs/HopkinsKPRSS17} this did not lead to unconditional lower bounds, it was used to obtain a characterization of sum-of-squares algorithms in terms of spectral algorithms, which we discuss further in \cref{sec:spectral}.

\subsection[Example: k-clique]{Example: \kclique}

In the remainder of this section, we will work out the pseudocalibration construction for the \kclique problem (see \cref{example:kclique} for a definition).
We'll follow the pseudocalibration recipe laid out in \cref{eq:pseudocalib}.
\paragraph{The null and structured distributions.} Recall that $\nulld$ is the uniform distribution over the hypercube $\{\pm 1\}^{\binom{[n]}{2}}$, corresponding to $\gnhalf$.
For $\jfdist$ we use the joint distribution over tuples of instance and hidden variables $(y^*,x^*)$ described in \cref{example:kclique}, with a small twist designed to ease calculations: Rather than sampling $x^*$ from $\pi$ the uniform distribution over the indicators $\Ind_S \in \bits^n$ for $|S| = k$, we sample $x^*$ by choosing every coordinate to be $1$ with probability $\tfrac{2k}{n}$, and $0$ otherwise.

\paragraph{Pseudomoments.}
Instead of describing the pseudodensity $\bar\mu$, it will be more convenient for us to work with the  {\em pseudomoments}.
So for each monomial $x^{A}$ where the multiset $A \subset [n]$ has cardinality at most $d$, we will directly define the function $\pE_{\bar\mu_y}[x^{A}]: \{\pm 1\}^{\binom{[n]}{2}} \to \R$.
For convenience, and to emphasize the dependence on $y$, we will equivalently write $\pE[x^{A}](y)$.
\medskip

Let $E^{\le D}$ be the set of subsets of edges with cardinality at most $D$.
Following the pseudocalibration recipe from \cref{eq:pseudocalib}, we project into the span of low-degree polynomials in $y$ using the monomial basis: for each $\alpha \in \cA$ we will compute the Fourier coefficient
\[
    \E_{y\sim \nulld} \left[y^\alpha \cdot \pE[x^{A}](y)\right]
    = \sum_{y \in \{\pm 1\}^{E}}y^{\alpha}\cdot \Pr_{\nulld}\{y\}\cdot \E_{x \sim \sigma} x^A \cdot \bar \mu_y
    = \E_{(x,y) \sim \feasd}[y^{\alpha}x^A].
\]
The right-hand side can be simplified further.
For $(x,y) \sim \jfdist$, if any vertices of $A$ are not chosen to be in the clique, then $x^A$ is zero.
Similarly, if any edge $e \in \alpha$ has an endpoint not in the clique, then $y^{\{e\}}$ is independent of $y^{A\setminus \{e\}}$ and of expectation 0.
Thus, the expression is equal to the probability that all vertices of $\alpha$ and $A$, which we denote $v(\alpha) \cup A$, are contained in the clique:
\[
    \E_{(x,y) \sim \feasd}[x^{A} y^{\alpha}] = \Pr_{x \sim \feasd}[ x_i = 1, \, \forall i \in v(\alpha) \cup A] =  \left(\tfrac{2k}{n}\right)^{|v(\alpha)\cup A|}.
\]
Now expressing $\pE[x^A](y)$ via its Fourier decomposition, we have
\begin{align}
    \pE(y)[x^{A}]
    &= \sum_{\alpha \in E^{\le D}} \left(\tfrac{2k}{n}\right)^{|v(\alpha)\cup A|} \cdot y^{\alpha}.\label{eq:expan}
\end{align}

It is an exercise to verify that scaling (\cref{eq:scaling}) holds up to errors of $o(1)$ for the pseudodistribution given by these moments when $k^2 \ll n$ and $D \ll \log n$, since for such $n,k,D$, the projection $\|(\Pi_{0,D} - \Pi_{0,0})\mu_*\| \le o(1)$.\footnote{This is equivalent to checking that the variance of \cref{eq:expan} is $o(1)$ for $A = \emptyset$.}
In a similar way one can also verify that the $\cP$ constraints (\cref{eq:sat}) are satisfied (since the conditions of \cref{thm:pseudocalib-constraints} are met for the natural polynomial system for clique) and that the condition of \cref{conj:low-deg} holds.
In the following subsection, we will discuss at a high level \cite{DBLP:conf/focs/BarakHKKMP16}'s proof that the positive semidefiniteness constraint (\cref{eq:psdness}) holds.

\subsection{Positive-semidefiniteness of matrix polynomials}\label{sec:matrix-poly}

To prove that the pseudocalibrated function $\bar \mu$ is a valid pseudodistribution, it remains to show that $\bar\mu$ satisfies the PSDness constraint \cref{eq:psdness}.
From \cref{claim:squaresimpliespsd}, we have that this is equivalent to proving that the matrix $M_d(y)$ defined in \cref{eq:mom-mat} is positive-semidefinite with high probability over $y\sim \nulld$.
Here we discuss, at a very high level, the proof of this fact for the planted clique problem from \cite{DBLP:conf/focs/BarakHKKMP16}.

Let $S,T \subset [n]^{\le d/2}$ be multisets that index the rows and columns of $M_d$.
In \cref{eq:expan}, we have shown that the entries of $M_d$ have the form
\[
    [M_d(y)]_{S,T}
    = \pE[ x^{S \cup T}]
= \sum_{\alpha \in E^{\le D}} \left(\frac{2k}{n}\right)^{|v(\alpha) \cup S \cup T|} \cdot y^{\alpha},
\]
where $E^{\le D}$ is the set of all subsets of edge variables with cardinality at most $D$.
Each entry of $M_d$ is a degree-$D$ polynomial in the random variable $y$, and because $M_d$ does not correspond in a natural way to a sum of independent random matrices, we cannot apply black-box matrix concentration results to $M_d$.

Since $\alpha \in E^{\le D}$ corresponds to a subset of edge variables, it is natural to associate with each $\alpha,S,T$ a colored subgraph or {\em shape} $\sigma$ on $|v(\alpha) \cup S \cup T|$ vertices.
The shape $\sigma$ is a graph $H_{\sigma}$ with vertex set isomorphic to $v(\alpha) \cup S \cup T$, and edge set isomorphic to $\alpha$.
Further, vertices isomorphic to $S$ are assigned colors $L =\ell_1,\ldots,\ell_{|S|}$, and vertices isomorphic to $T$ are assigned colors $R =r_1,\ldots,r_{|T|}$ (note that a single vertex may receive more than one color).
For example, if $S = \{a,b,k\}$, $T = \{i,j,k\}$, and $\alpha = \{(a,i), (b,u), (u,j)\}$, we would have the corresponding shape

\begin{center}
\begin{tikzpicture}
    \node (a) at (0,1) [label=left:{$a$}] {};
    \node (b) at (0,0) [label=left:{$b$}] {};
    \node (k) at (1,-0.7) [label=below:{$k$}] {};
    \node (i) at (2,1) [label=right:{$i$}] {};
    \node (j) at (2,0) [label=right:{$j$}] {};
    \node (u) at (1,.3) [label=above:{$u$}] {};
    \fill [black] (a) circle[radius = 2pt];
    \fill [black] (b) circle[radius = 2pt];
    \fill [black] (k) circle[radius = 2pt];
    \fill [black] (i) circle[radius = 2pt];
    \fill [black] (j) circle[radius = 2pt];
    \fill [black] (u) circle[radius = 2pt];
    \draw [black] (a) -- (i);
    \draw [black] (b) -- (u) -- (j);
    \draw [dashed,rounded corners,NavyBlue] (-.1,1.1) --(0,1.3) -- (0.2,1.1) -- (0.4,0.1)-- (1.2,-0.6) -- (1.2,-0.9) -- (1,-0.9) -- (0.8,-0.8) -- (-0.15,-0.2) -- (-.2,0) -- (-.2,1) -- (-.1,1.1) ;
    \node (S) at (-0.8,-0.4) {${\color{NavyBlue} \mathbf S}$};
    \draw [dashed,rounded corners,ForestGreen] (2.1,1.1) --(2,1.3) -- (1.8,1.1) -- (1.6,0.1)-- (0.8,-0.6) -- (0.8,-0.9) -- (1,-0.9) -- (1.2,-0.8) -- (2.15,-0.2) -- (2.2,0) -- (2.2,1) -- (2.1,1.1) ;
    \node (T) at (2.8,-.4) {${\color{ForestGreen} \mathbf T}$};
    \draw [thick,->] (4,0.3) -- (6,0.3);
    \node (l1) at (8,1) [label=left:{$\ell_1$}] {};
    \node (l2) at (8,0) [label=left:{$\ell_2$}] {};
    \node (c) at (9,-0.7) [label=below:{$\ell_3,r_3$}] {};
    \node (r1) at (10,1) [label=right:{$r_1$}] {};
    \node (r2) at (10,0) [label=right:{$r_2$}] {};
    \node (v) at (9,.3) {};
    \fill [black] (l1) circle[radius = 2pt];
    \fill [black] (l2) circle[radius = 2pt];
    \fill [black] (c) circle[radius = 2pt];
    \fill [black] (r1) circle[radius = 2pt];
    \fill [black] (r2) circle[radius = 2pt];
    \fill [black] (v) circle[radius = 2pt];
    \draw [black] (l1) -- (r1);
    \draw [black] (l2) -- (v) -- (r2);
    \draw [dashed,rounded corners,NavyBlue] (7.9,1.1) --(8,1.3) -- (8.2,1.1) -- (8.4,0.1)-- (9.2,-0.6) -- (9.2,-0.9) -- (9,-0.9) -- (8.8,-0.8) -- (7.85,-0.2) -- (7.8,0) -- (7.8,1) -- (7.9,1.1) ;
    \node (L) at (7.2,-.4) {${\color{NavyBlue} \mathbf L}$};
    \draw [dashed,rounded corners,ForestGreen] (10.1,1.1) --(10,1.3) -- (9.8,1.1) -- (9.6,0.1)-- (8.8,-0.6) -- (8.8,-0.9) -- (9,-0.9) -- (9.2,-0.8) -- (10.15,-0.2) -- (10.2,0) -- (10.2,1) -- (10.1,1.1) ;
    \node (R) at (10.8,-.4) {${\color{ForestGreen} \mathbf R}$};
    \node (label1) at (1,1.8) {$S,T,\alpha$};
    \node (label2) at (9,1.8) {$\text{shape}(S,T,\alpha)$};
\end{tikzpicture}
\end{center}

For each such shape $\sigma$, we define the $n^{\le d} \times n^{\le d}$ {\em $\sigma$-matrix polynomial} $M_{\sigma}(y)$, so that for $S,T \in [n]^{\le d}$,
\begin{equation}
    [M_{\sigma}(y)]_{S,T} = \sum_{\alpha \in E^{\le D}} \Ind[\text{shape}(S,T,\alpha) = \sigma]\cdot y^{\alpha}.\label{eq:shape-sum}
\end{equation}
Or alternatively for $S,T,L \subset [n]$, let $H_{\sigma}(S,T,U)$ be the labeled copy of $H_{\sigma}$ in which $R$ is labeled with $S$, $L$ with $T$, and the remaining vertices with $U$, and every vertex of $H_{\sigma}$ receives a unique, single label.
Then \cref{eq:shape-sum} is equivalent to summing over products of edges for all valid labelings of the uncolored vertices of $H_{\sigma}$:
\[
    [M_{\sigma}(y)]_{S,T} = \sum_{U \in \binom{[n]\setminus (S \cup T)}{|V(H) \setminus (L\cup R)|}} \Ind[H_{\sigma}(S,T,U)\text{ valid }]\cdot y^{E(H_{\sigma}(S,T,U))}.
\]

The matrices $\{M_{\sigma}\}$ form a natural basis for expressing $M_d(y)$:
\[
    M_d(y) = \sum_{\sigma} \left(\frac{2k}{n}\right)^{|v(\sigma)|} \cdot M_{\sigma}(y)\mper
\]

In \cite{DBLP:conf/focs/BarakHKKMP16}, the authors characterize the spectrum of $M_{\sigma}(y)$.
Incredibly, the spectral properties of the $\sigma$-matrix polynomials  determined by the connectivity of $H_\sigma$.
\begin{theorem}[\cite{DBLP:conf/focs/BarakHKKMP16,DBLP:conf/approx/MedarametlaP16}]\label{thm:shape-spectrum}
    Suppose that $H_{\sigma}$ has $t = O(\log n)$ vertices, and that $H_{\sigma}$ has exactly $p$ vertex-disjoint paths from $L \setminus R$ to $R \setminus L$, and that $|R \cap L| = c$.
    Then with high probability over $y \sim \{\pm 1\}^N$,
    \[
	\left\|M_{\sigma}(y)\right\| \le 2^{O(t)} (\log n)^{O(t +p- c)}\cdot n^{\frac{t - p - c}{2}}.
    \]
\end{theorem}

In order to characterize the spectrum of $M_d(y)$, it does not suffice to understand the spectrum of each $M_{\sigma}$ individually; one must account for the interactions of the spectra of the $M_{\sigma}$.
This is challenging because different shapes exhibit very different spectral characteristics.
For example, for $\sigma_1$ the shape given by two horizontal parallel paths of length $1$, the matrix $M_{\sigma_1}$ has spectral norm of magnitude $\tilde O(n)$.
\begin{center}
\begin{tikzpicture}
    \node (a) at (0,1) [label=left:{$\ell_1$}] {};
    \node (b) at (0,0) [label=left:{$\ell_2$}] {};
    \node (i) at (2,1) [label=right:{$r_1$}] {};
    \node (j) at (2,0) [label=right:{$r_2$}] {};
    \fill [black] (a) circle[radius = 2pt];
    \fill [black] (b) circle[radius = 2pt];
    \fill [black] (i) circle[radius = 2pt];
    \fill [black] (j) circle[radius = 2pt];
    \draw [black] (a) -- (i);
    \draw [black] (b) -- (j);
    \draw [dashed,rounded corners,NavyBlue] (-.2,1.2) --(0,1.3) -- (0.2,1.2)  -- (0.2,-0.2) -- (0,-.3) -- (-.2,-.2) -- (-.2,1.2) ;
    \node (S) at (-0.9,0.4) {${\color{NavyBlue} \mathbf L}$};
    \draw [dashed,rounded corners,ForestGreen] (2.2,1.2) --(2,1.3) -- (1.8,1.2)  -- (1.8,-0.2) -- (2,-.3) -- (2.2,-.2) -- (2.2,1.2) ;
    \node (T) at (2.9,0.4) {${\color{ForestGreen} \mathbf R}$};
    \node (label1) at (1,1.8) {$\sigma_1$};
    \node (a) at (7,1) [label=left:{$\ell_1$}] {};
    \node (b) at (7,0) [label=left:{$\ell_2$}] {};
    \node (i) at (9,1) [label=right:{$r_1$}] {};
    \node (j) at (9,0) [label=right:{$r_2$}] {};
    \fill [black] (a) circle[radius = 2pt];
    \fill [black] (b) circle[radius = 2pt];
    \fill [black] (i) circle[radius = 2pt];
    \fill [black] (j) circle[radius = 2pt];
    \draw [black] (a) -- (b);
    \draw [black] (i) -- (j);
    \draw [dashed,rounded corners,NavyBlue] (6.8,1.2) --(7,1.3) -- (7.2,1.2)  -- (7.2,-0.2) -- (7,-.3) -- (6.8,-.2) -- (6.8,1.2) ;
    \node (S) at (6.1,0.4) {${\color{NavyBlue} \mathbf L}$};
    \draw [dashed,rounded corners,ForestGreen] (9.2,1.2) --(9,1.3) -- (8.8,1.2)  -- (8.8,-0.2) -- (9,-.3) -- (9.2,-.2) -- (9.2,1.2) ;
    \node (T) at (9.9,0.4) {${\color{ForestGreen} \mathbf R}$};
    \node (label1) at (8,1.8) {$\sigma_2$};
\end{tikzpicture}
\end{center}
On the other hand, for $\sigma_2$ the shape given by two vertical parallel paths of length $1$,the matrix $M_{\sigma_2}$ has spectral norm  $\Omega(n^2)$.\footnote{Another way to see that this is true without \cref{thm:shape-spectrum} is that if $A$ is the signed adjacency matrix of the random graph, then $\|A\|^2_F = n(n-1)$ and with high probability $\|A\| = n$, and (excluding entries corresponding to $S,T$ with nontrivial intersection) $M_{\sigma_1} = A \otimes A$, while $M_{\sigma_2} = v(A)v(A)^\top$ where $v(A)$ is the $n^2 \times 1$ reshaping of $A$.}

To prove that $M_d(y)$ is indeed PSD with high probability, \cite{DBLP:conf/focs/BarakHKKMP16} employ a delicate iterative charging scheme, partitioning the $\sigma$ into groups according to the number of disjoint paths from $L\setminus R$ to $R\setminus L$ and charging the spectral norm of each such group to the positive-semidefinite shapes.
We will not give further details, and refer the interested reader to \cite{DBLP:conf/focs/BarakHKKMP16}.

\paragraph{Symmetric matrix polynomials.}
The matrix $M_d(y)$ is not completely devoid of structure: it is what we call a {\em symmetric matrix polynomial}, since the matrix remains fixed under permutations of the elements of $[n]$.
Specifically, if we let $\pi \in \cS_n$, and if $\pi(M)$ and $\pi(y)$ are the natural action induced by permutations of the vertices of $G$ on the rows/columns of $M$ and edges in $y$, then
\[
    \pi(M_{\sigma}(\pi(y))) = M_{\sigma}(y).
\]
For a familiar example of a symmetric matrix polynomial, consider the adjacency matrix of $G$ with entries expressed as degree-$1$ polynomials of $y$.

One compelling question is whether the symmetry of $M_d$ can be useful in characterizing its spectrum.
\begin{question}
    Suppose that $A(y)$ is a symmetric matrix polynomial.
	What are sufficient conditions on $A$ such that $A(y) \succeq 0$ with high probability over $y \sim \nulld$?
\end{question}
For some classes of symmetric matrices, such as association scheme matrices, the above question is fully answered (see e.g. \cite{godsil2010association}).
There is hope that if this question is answered in greater generality, it will lead to a ``book proof'' of the pseudocalibration method.

%% file: content/spectral.tex
\section{Connection to spectral algorithms}
\label{sec:spectral}

Sum-of-squares SDPs yield a systematic framework that captures and generalizes a loosely defined class of algorithms often referred to as \emph{spectral algorithms}.
We say that an algorithm is a ``spectral algorithm'' if on input $y$ the algorithm constructs a matrix $M(y)$ that can be easily computed from $y$, whose eigenvalues or eigenvectors manifestly yield a solution to the problem at hand.\footnote{In other contexts, ``spectral algorithms'' may sometimes describe algorithms that also modify $M(y)$ as the algorithm proceeds; for simplicity and because our main result will be an {\em equivalence} between SoS and spectral algorithms, we consider only this narrower class.}
We will give a more concrete definition for the notion of a spectral algorithm a little later in this section.

Although spectral algorithms are typically subsumed by sum-of-squares SDPs, they tend to be simpler to implement and more efficient.
Furthermore, in many cases such as \kclique \cite{MR1662795-Alon98}
 and \tensordecomposition\ \cite{harshman1970foundations}, the first algorithms discovered for the problem were spectral.
From a theoretical standpoint, spectral algorithms are much simpler to study and could serve as stepping stones to understanding the limits of sum-of-squares SDPs.

In the worst case, sum-of-squares SDPs often yield strictly better guarantees than corresponding spectral algorithms.
For instance, the Goemans-Williamson SDP (corresponding to an SoS SDP of degree $2$) yields a 0.878 approximation for \maxcut \cite{MR1412228-Goemans95}, and has no known analogues among spectral algorithms.
Contrary to this, in many random settings, the best known SoS algorithms yield guarantees that are no better than the corresponding spectral algorithms.
Recent work explains this phenomenon by showing an equivalence between spectral algorithms and their sum-of-squares counterparts for a broad family of problems \cite{DBLP:conf/focs/HopkinsKPRSS17}.

\subsection{Spectral algorithms and sum-of-squares proofs}

Before formally stating the equivalence of SoS and spectral algorithms from \cite{DBLP:conf/focs/HopkinsKPRSS17}, we first demonstrate that SoS often captures spectral algorithms.
Let us begin by considering a classic example of a simple spectral algorithm for the \kclique problem.
\paragraph{Spectral algorithms for \kclique.} In a graph $G = (V,E)$ with adjacency matrix $A_G$, if a subset $S\subset V$ of $k$ vertices forms a clique then,
\[ \Iprod{ \One_S, \left(A_G - \tfrac{1}{2}J\right) \One_S} = \frac{k(k-1)}{2} \mper  \]
where $J \in \R^{n \times n}$ denotes the $n \times n$ matrix consisting of all ones.
On the other hand, we can upper bound the left-hand side by
\[\Iprod{ \One_S, \left(A_G - \tfrac{1}{2}J\right) \One_S} \leq \norm{\One_S}_2^2\cdot \Norm{A_G - \tfrac{1}{2}J}_{op} = k \cdot \lmax\left({A_G - \tfrac{1}{2}J}\right)  \mper \]
The eigenvalue of this matrix thus certifies an upper bound on the size of the clique $k$, namely,
\[ k \leq 2 \lmax\left({A_G - \tfrac{J}{2}}\right) +2 \mper\]
In particular, for a graph $G$ drawn from the null distribution $\gnhalf$, the matrix $A_G - \frac{1}{2}J$ is a random symmetric matrix whose off-diagonal entries are i.i.d uniform over $\{\pm \frac{1}{2}\}$.
By a classical result in random matrix theory (Bai-Yin's law, see e.g. \cite{tao2012topics}), we will have that $\lmax \left(A_G - \frac{1}{2}J\right) = O(\sqrt{n})$ with high probability.
Thus one can certify an upper bound of $O(\sqrt{n})$ on the size of the clique in a random graph drawn from $\gnhalf$ by computing the largest eigenvalue of the associated matrix valued function $P(G) = A_G - \frac{1}{2} J$.

This algorithm also gives a degree-$2$ sum-of-squares proof; if $x \in \{0,1\}^n$ are the indicator variables for vertex membership in the clique polynomial system $\cA$ (as in \cref{example:kclique}), then we have that the clique size $k = \sum_i x_i$, and
\begin{align*}
    \cA \sststile{}{} \left(\textstyle\sum_i x_i\right)^2
    &= \sum_{i,j} x_i x_j - 2x_ix_j \cdot \Ind[(i,j) \not \in E(G)]\\
    &= x^\top(2A_G + 2\Id - J)x\\
    &= x^\top\left(2\cdot \lambda_{\max}(A_G + \Id - \tfrac{1}{2}J)\cdot \Id\right)x + \textstyle\sum_{j} s_j(x)^2\\
    &= \|x\|^2\cdot (2\cdot \lambda_{\max}(A_G - \tfrac{1}{2}J) + 2) + \textstyle\sum_{j} s_j(x)^2\\
    &= \left(\textstyle\sum_{i} x_i\right)\cdot (2\cdot \lambda_{\max}(A_G - \tfrac{1}{2}J) + 2)+ \textstyle\sum_{j} s_j(x)^2,
\end{align*}
where in the first line we have used that $\{x_ix_j = 0\}_{(i,j) \not\in E(G)} \in \cA$, to derive the third line we have used that for a symmetric matrix $M$, $\lambda_{\max}(M)\cdot \Id - M$ is positive-semidefinite and therefore its eigendecomposition gives that $x^\top \left(\lambda_{\max}(M)\cdot \Id - M\right)x$ is a sum-of-squares.
To derive the last line we have used that $ \{x_i^2 = x_i\}_{i\in[n]}\in \cA$.
\paragraph{Spectral algorithms for injective tensor norm.}
Now we will see another example of a more complex spectral algorithm which is captured by sum-of-squares.
Recall that the injective tensor norm (see \cref{ex:tensorpca}) of a symmetric $4$-tensor $\bT \in \R^{[n] \times [n] \times [n] \times [n]}$ is given by $\max_{\norm{x} \leq 1} \iprod{x^{\otimes 4}, \bT}$.
The injective tensor norm $\norminj{\bT}$ is computationally intractable in the worst case \cite{MR3144915-Hillar13}.
We will now describe a sequence of spectral algorithms
that certify tighter bounds for the injective tensor norm of a tensor $\bT$ drawn from the null distribution (with entries drawn i.i.d from $\cN(0,1)$).
Let $T = T_{\{1,2\},\{3,4\}}$ denote the $n^2 \times n^2$ matrix obtained by reshaping the tensor $\bT$.
Then,
\[
\norminj{\bT}
= \argmax_{ \norm{x}_2 \leq 1 } \iprod{\bT,x^{\otimes 4}}
= \argmax_{ \norm{x}_2 \leq 1 } \iprod{x^{\otimes 2}, T x^{\otimes 2}} \leq \lmax(T) \]
Thus $\lmax(T)$ is a spectral upper bound on $\norminj{\bT}$.
Since each entry of $\bT$ is drawn independently from $\cN(0,1)$,  we have from the Bai-Yin law that $\lmax(T) \leq  O(n)$ with high probability \cite{tao2012topics}.
We also recall that the injective norm of a random $\cN(0,1)$ tensor $T$ is at most $O(\sqrt{n})$ with high probability \cite{doi:10.1002/cpa.21422,MontanariR14}
.
Taking these facts together, $\lmax(T)$ certifies an upper bound that is $O(\sqrt{n})$-factor approximation to $\norminj{\bT}$.

We will now describe a sequence of improved approximations to the injective tensor norm via spectral methods, which also yield an analysis of the SoS SDP for \tensorpca.
Fix a positive integer $k \in \N$.  The polynomial $\bT(x) = \iprod{x^{\otimes 4}, \bT}$ can be written as,
\[
\bT(x)
= \iprod{x^{\otimes 2}, T x^{\otimes 2}} = \iprod{x^{\otimes 2k}, T^{\otimes k} x^{\otimes 2k}}^{1/k} \mper
\]
The tensored vector $x^{\otimes 2k}$ is symmetric, and is invariant under permutations of its modes.
Let $\Sigma_{2k}$ denote the set of all permutations of $ \{1,\ldots, 2k\}$.  For a permutation $\Pi \in \Sigma_{2k}$ and a $2k$-tensor $A \in \R^{[n]^{2k}}$, let $\Pi \circ A$ denote the $2k$-tensor obtained by applying the permutation $\Pi$ to the modes of $A$.
By averaging over all permutations $\Pi, \Pi' \in \Sigma_{2k}$, we can write
\begin{align}
\bT(x)
    & = \left( \E_{\Pi, \Pi' \in \Sigma_{2k} } \Iprod{ \Pi \circ x^{\otimes 2k}, T^{\otimes k} (\Pi'  \circ x^{\otimes 2k})} \right)^{1/k} \nonumber \\
      & = \left(  \Iprod{ x^{\otimes 2k}, \left(\E_{\Pi, \Pi' \in \Sigma_{2k} } \Pi \circ T^{\otimes k} \circ \Pi' \right)  x^{\otimes 2k}} \right)^{1/k} \nonumber \\
      & \leq \lmax\left( \E_{\Pi, \Pi' \in \Sigma_{2k} } \Pi \circ T^{\otimes k} \circ \Pi' \right)^{1/k} \cdot \norm{x}_2^4 \mper \label{eq:symmetrizedmatrix}
\end{align}
Therefore for every $k \in \N$, if we denote
\[ P_k(\bT) \defeq \E_{\Pi, \Pi' \in \Sigma_{2k} } \Pi \circ T^{\otimes k} \circ \Pi' \]
then $\norminj{\bT} \leq \lmax(P_k(\bT))^{1/k}$.

The entries of $P_k(\bT)$ are degree-$k$ polynomials in the entries of $\bT$.
For example, a generic entry of $P_{2}(\bT)$ looks like,
\[ P_2(\bT)_{ijk\ell,abcd} = \frac{1}{(4!)^2} \cdot \left(T_{ijab} \cdot T_{k \ell c d} + T_{ij a c} \cdot T_{k \ell b d} + T_{ij a d} \cdot T_{k \ell b c} + \cdots  \right) \mcom
\]
where we sum over the $(4!)^2$ pairs of permutations of $i,j,k,\ell$ and $a,b,c,d$.
Thus a typical entry of $P_k(\bT)$ with no repeated indices is an average of a super-exponentially large (in $k$) number of i.i.d. random variables.
We will call this number $N_k$ for convenience.

When $k \ll \sqrt{n}$, a typical entry of $P_k(\bT)$ contains no repeated indices, and this implies that the variance of a typical entry of $P_k(\bT)$ is equal to $\frac{1}{N_k}$.
For the moment, let us assume that the spectrum of $P_k(\bT)$ has a distribution that is similar to that of a random matrix with i.i.d. Gaussian entries with variance $\frac{1}{N_k}$.
Then, $\lmax(P_k(\bT)) \leq O(n^{k} \cdot \frac{1}{N_k^{1/2}}) $ with high probability, certifying that $\norminj{\bT} \leq \frac{n}{N_k^{1/2k}}$.
On accounting for the symmetries of $\bT$, it is not difficult to see that $N_k = k!\left( \frac{1}{2^k} \tfrac{2k!}{k!} \right)^2 \gg (k!)^3$.
Consequently, as per this heuristic argument, $\lmax(P_k(\bT))$ would certify an upper bound of $\norminj{\bT} \leq O(\frac{n}{k^{3/2}})$.
Unfortunately, the entries of $P_k(\bT)$ are not independent random variables and not all entries of $P_k(T)$ are typical as described above.
Although the heuristic bound on $\lmax(P_k(\bT))$ is not quite accurate, a careful analysis via the trace method shows that the upper bound $\lmax(P_k(\bT))^{1/k}$ decreases polynomially in $k$ \cite{DBLP:conf/approx/BhattiproluGL17, DBLP:conf/stoc/RaghavendraRS17}.
\begin{theorem} \label{thm:tensorBound} \cite{DBLP:conf/approx/BhattiproluGL17}
	For $4 \leq k \leq n$, if $\bT$ is a symmetric $4$-tensor with i.i.d. entries from a subgaussian measure then \[\lmax(P_k(\bT))^{1/k} \leq \tilde{O}\left(\frac{n}{k^{1/2}} \right)  \]
then with probability $1 - o(1)$.  Here $\tilde{O}$ notation hides factors polylogarithmic in $n$.
\end{theorem}
Thus the matrix polynomial $P_k(\bT)$ yields a $n^{O(k)}$-time algorithm to certify an upper bound of $\tilde{O}(n/k^{1/2})$ on the injective tensor norm of random $4$-tensors with Gaussian entries.

\paragraph{Spectral algorithms from Sum-of-Squares analyses.}
In fact, this spectral algorithm was discovered in the context of analyzing SoS refutation algorithms, and the upper bound certificate produced by the above spectral algorithm can again be cast as a degree $4k$ sum-of-squares proof.
In particular, if $\lmax(P_k(\bT)) \leq \tau$ for some tensor $\bT$ and $\tau \in \R$ then,
\begin{align*}
 \tau - \bT(x)^{k} & =  \tau \norm{x}_2^{4k} - \iprod{x^{\otimes 2k}, P_k(\bT) x^{\otimes 2k}} + \tau( 1 - \norm{x}_2^{4k}) \\
  &  =\iprod{x^{\otimes 2k}, (\tau \cdot \Id - P_k(\bT)) x^{\otimes 2k}} + \tau(1 - \norm{x}_2^{4k})\\
  &   =\iprod{x^{\otimes 2k}, (\tau \cdot \Id - P_k(\bT)) x^{\otimes 2k}} + (1-\norm{x}_2^2) \left( \tau \cdot \sum_{i = 0}^{2k-1} \norm{x}_2^{2 i}\right)\\  & = \sum_{j} s_j^2(x) + (1-\norm{x}_2^2) \left( \tau \cdot \sum_{i = 0}^{2k-1} \norm{x}_2^{2 i}\right),
\end{align*}
The final step in the calculation again uses the fact that if a matrix $M \succeq 0$, then the polynomial $\iprod{x^{\otimes 2k}, M x^{\otimes 2k}}$ is a sum-of-squares $\sum_{j} s_j^2(x)$.
Therefore, the degree-$4k$ sum-of-squares obtains an approximation guarantee that is no worse than the somewhat ad hoc spectral algorithm described above.

This is a recurrent theme where the sum-of-squares SDP yields a unified and systematic algorithm that subsumes a vast majority of more ad hoc algorithms.
It also exemplifies the trend of SoS inspiring new spectral algorithms which take advantage of polynomial identities and problem symmetry to improve upon simpler spectral algorithms (see also \cite{DBLP:conf/colt/HopkinsSS15,DBLP:conf/focs/AllenOW15,DBLP:conf/approx/BhattiproluGL17,DBLP:conf/stoc/RaghavendraRS17}).
There is a line of work in which the spectral certificates used by SoS for refutation are modified and compressed to give efficient, lightweight spectral algorithms that run in subquadratic or near-linear time; we refer the interested reader to \cite{DBLP:conf/stoc/HopkinsSSS16,MontanariS16,DBLP:conf/colt/SchrammS17}.

\paragraph{Refuting Random CSPs.}
The basic scheme used to upper bound the injective tensor norm (see \cref{eq:symmetrizedmatrix}) can be harnessed towards refuting random constraint satisfaction problems (CSPs).
Fix a positive integer $k \in \N$.
In general, a random $k$-CSP instance consists of a set of variables $V$ over a finite domain, and a set of randomly sampled constraints each of which is on a subset of at most $k$ variables.
The problem of refuting random CSPs has been extensively studied for its numerous connections and applications \cite{MR2121179-Feige02,DBLP:journals/eccc/ECCC-TR02-003,DBLP:conf/stoc/DanielyLS14,DBLP:conf/innovations/BarakKS13,crisanti20023}.
For the sake of concreteness, let us consider the example of random \fourxor.
\begin{example}[\fourxor]
In the \fourxor problem, the input consists of a homogeneous degree-$4$ linear system of $m$ equations over $n$ variables $\{X_1,\ldots,X_n\}$ in $\F_2$.
A random \fourxor instance is one where each equation is sampled uniformly at random (avoiding repetition).
For $m \gg n$, with high probability over the choice of the constraints, no assignment satisfies more than a $\frac{1}{2} + o(1)$ fraction of constraints.

    To formulate a polynomial system, we will use the natural $\{\pm 1\}$-encoding of $\F_2$, i.e., $x_i = 1 \iff X_i = 0$ and $x_i = -1 \iff X_i = 1$.
   An equation of the form $X_i + X_j + X_k + X_\ell = 0/1$ translates in to $x_i x_j x_k x_\ell = \pm 1$.
   We can specify the instance using a symmetric $4$-tensor $\{\bT_{ijk \ell}\}_{i,j,k,\ell \in \binom{[n]}{4}}$, with $\bT_{ijk \ell} = \pm 1$ if we have the equation $x_i x_jx_k x_\ell = \pm 1$, and $\bT_{ijk} = 0$ otherwise.
   To certify that no assignment satisfies more than $\epsilon m$ constraints, we will need to refute the following polynomial system.
   \begin{align}
  \left\{ x_i^2 -1 \right\}_{i \in [n]}  \qquad \text{ and } \qquad \left\{ \iprod{\bT, x^{\otimes 4}} \geq \epsilon \cdot m\right\}
   \end{align}
    A refutation for this system with $\eps < 1-\eta$ for $\eta$ independent of $m,n$ is called a {\em strong refutation}.
   This system is analogous to the injective tensor norm, except the maximization is over the Boolean hypercube $x \in \sbits^n$, as opposed to the unit ball.
   Unlike the case of random Gaussian tensors, the tensor $\bT$ of interest in \fourxor is a sparse tensor with about $n^{1+ o(1)}$ non-zero entries.
   While this poses a few technical challenges, the basic schema from \cref{eq:symmetrizedmatrix} can still be utilized to obtain the following strong refutation algorithm.
   \begin{theorem} \cite{DBLP:conf/stoc/RaghavendraRS17}
   For all $\delta \in [0,1)$, the degree $n^{\delta}$ sum-of-squares SDP can strongly refute random \fourxor instances with $m > \tilde{\Omega}(n^{2-\delta})$ with high probability.
   \end{theorem}
   The refutation algorithm for XOR can be used as a building block to obtain sum-of-squares refutations for all random $k$-CSPs \cite{DBLP:conf/stoc/RaghavendraRS17}.
   Moreover, these bounds on the degree of sum-of-squares refutations tightly match corresponding lower bounds for CSPs shown in \cite{DBLP:conf/stoc/KothariMOW17, MR3388187-Barak15}.
\end{example}

\subsection{Equivalence of spectral algorithms and sum-of-squares refutations}
The algorithms described above will serve as blueprints for a class of spectral algorithms that will characterize the power of SoS SDPs.

\paragraph{Defining spectral algorithms.}
Here, we will consider spectral algorithms for distinguishing problems.
Recall that in a distinguishing problem, the input consists of a sample $y$ drawn from one of two distributions, say a structured distribution $\planteddist$ or a null distribution $\randomdist$, and the algorithm's goal is to identify the distribution the sample is drawn from.
We think of samples from the structured distribution $\planteddist$ as having an underlying hidden structure, while samples from the null distribution $\randomdist$ typically do not.

A \emph{spectral} algorithm $\cA$ for the distinguishing problem proceeds as follows.
Given an instance $y$, the algorithm $\cA$ computes a matrix $P(y)$ whose entries are given by low-degree polynomials in $y$, such that $\lmax(P(y))$ indicates whether $y \sim \planteddist$ or $y \sim \randomdist$.
\begin{definition} \label{def:spectralalg} (Spectral Algorithm)
	A spectral algorithm $\cA$ consists of a matrix valued polynomial $P: \cP \to \R^{N \times N}$.  The algorithm $\cA$ is said to distinguish between samples from structured distribution $\planteddist$ and a null distribution $\randomdist$  if,
	\[  \E_{y \sim \planteddist} \lmax^+(P(y)) \gg \E_{y \sim \randomdist} \lmax^+(P(y)) \]
	where $\lmax^+(M) \defeq \max(\lmax(M), 0)$ for a matrix $M$.
\end{definition}
In general, a spectral algorithm could conceivably use the entire spectrum of the matrix $P(y)$ instead of the largest eigenvalue, and perform some additional computations on the spectrum.
However, a broad range of spectral algorithms can be cast into this framework and as we will describe in this section, this restricted class of spectral algorithms already subsumes the sum-of-squares SDP in a wide variety of settings.

Spectral algorithms as defined in \cref{def:spectralalg} are a simple and highly structured class of algorithms, in contrast to algorithms for solving a sum-of-squares SDP.
The feasible region for a sum-of-squares SDP is the intersection of the positive semidefinite cone with polynomially many constraints, some of which are equality constraints.
Finding a feasible solution to the SDP involves an iterated sequence of eigenvalue computations.
Furthermore, the feasible solution returned by the SDP solver is by no-means guaranteed to be a low-degree function of the input instance.
On the other hand, a spectral algorithm involves exactly one eigenvalue computation of a matrix whose entries are low-degree polynomials in the instance.
In spite of their apparent simplicity, we will now argue that spectral algorithms are no weaker than sum-of-squares SDPs for a wide variety of estimation problems.

\paragraph{Robust inference.}
Many estimation problems share the "robust inference" property.
Specifically, the structured distributions underlying these estimation problems are such that a randomly chosen subsampling of the instance is sufficient to recover a non-negligible fraction of the planted structure.
For example, consider the structured distribution $\planteddist$ for the \kclique problem.  A graph $G \sim \planteddist$ consists of a $k$-clique embedded in to an \Erdos-\Renyi random graph.
Suppose we subsample an induced subgraph $G'$ of $G$, by randomly sampling a subset $S \subset V$ of vertices of size $\card{S} = \delta \card{V}$.
With high probability, $G'$ contains $\Omega(\delta \cdot k)$ of the planted clique in $G$.
Therefore, the maximum clique in $G'$ yields a clique of size $\Omega(\delta \cdot k)$ in the original graph $G$.
This is an example of the \emph{robust inference} property, where a random subsample $G'$ can reveal non-trivial structure in the instance.

Though the subsample does not determine the planted clique in $G$, the information revealed is substantial.
For example, as long as $\delta \cdot k \gg 2\log{n}$, observing $G'$ allows us to distinguish whether $G$ is sampled from the structured distribution $\planteddist$ or the null distribution $\randomdist$.
Moreover, the maximum clique in $G'$ can be thought of as a feasible solution to a relaxed polynomial system where the clique size sought after is $\delta \cdot k$, instead of $k$.

Let $\cP$ denote a polynomial system defined on instance variables $y \in \R^N$ and in solution variables $x \in \R^n$.
We define the {\em subsampling distribution} $\subdist$ to be a probability distribution over subsets of instance variables $[N]$.
Given an instance $y \in \R^N$, a subsample $z$ can be sampled by first picking $S \sim \subdist$ and setting $z = y_S$.
Let $\cI$ denote the collection of all instances, and $\cI_{\downarrow}$ denote the collection of all sub-instances.
\begin{definition}[Robust inference]
	A polynomial system $\cP$ is $\epsilon$-robustly inferable with respect to a subsampling distribution $\subdist$ and a structured distribution $\planteddist$, if there exists a map $\zeta: \cI_{\downarrow} \to \R^n$ such that,

	\[ \Pr_{\substack{y \sim \planteddist\\ S \sim \subdist}}[ \zeta(y_{S}) \text{ is feasible for } \cP] \geq 1- \epsilon \]
\end{definition}
The robust inference property arises in a broad range of estimation problems including stochastic block models, \densestksubgraph, \tensorpca, sparse PCA and random CSPs (see \cite{DBLP:conf/focs/HopkinsKPRSS17} for a detailed discussion).
The existence of the robust inference property has a dramatic implication for the power of low-degree sum-of-squares SDPs: they are no more powerful than spectral algorithms.
This assertion is formalized in the following theorem.
\begin{theorem}[\cite{DBLP:conf/focs/HopkinsKPRSS17}]\label{thm:spectralSDP}
Suppose $\cP = \{p_i(x,y) \geq 0 \}_{i \in [m]}$ is a polynomial system with degree $d_x$ and $d_y$ over $x$ and $y$ respectively.  Fix $B \geq d_x \cdot d_y \in \N$.
If the degree-$d$ sum-of-squares SDP relaxation can be used to distinguish between the structured distribution $\planteddist$ and the null distribution $\randomdist$, namely,
\begin{itemize}
	\item For $y \sim \planteddist$, the polynomial system $\cP$ is not only satisfiable, but is $1/n^{8B}$-robustly inferable with respect to a sub-sampling distribution $\subdist$.

	\item For $y \sim \randomdist$, the polynomial system $\cP$ is not only infeasible but admits a degree-$d$ sum-of-squares refutation with numbers bounded by $n^B$ with probability at least $1 - 1/n^{8B}$. \Pnote{}

\end{itemize}
Then, there exists a degree-$2D$ matrix polynomial $Q: \cI \to \R^{[n]^{\leq d} \times [n]^{\leq d}}$ such that,
\[ \frac{\E_{y \sim \planteddist}[\lmax^+(Q(y))]}{\E_{y \sim \randomdist}[\lmax^+(Q(y))]} \geq n^{B/2}\]
where $D \in \N$ is smallest integer such that for every subset $\alpha \subset [N]$ with $\card{\alpha} \geq D - 2d_x d_y$, $\Pr_{S \sim \subdist}[\alpha \subseteq S] \leq \frac{1}{n^{8B}}$.
\end{theorem}
The degree $D$ of the spectral distinguisher depends on the sub-sampling distribution.  Intuitively, the more robustly inferable (a.k.a inferable from smaller subsamples) the problem is, the smaller the degree of the distinguisher $D$.  For the \kclique problem with a clique size of $n^{1/2 -\epsilon}$,  we have $D = O(d/\epsilon)$.  For the typical parameter settings of random CSPs, community detection and densest subgraph we have $D = O(d \log n)$ (see \cite{DBLP:conf/focs/HopkinsKPRSS17} for details).

From a practical standpoint, the above theorem shows that sum-of-squares SDPs can often be replaced by their more efficient spectral counterparts.\footnote{We comment however that the matrix polynomial $Q(y)$ is non-constructive and non-uniform, and arises as the dual object of a exponentially-sized convex program. For this reason, the theorem does not automatically give efficient spectral algorithms matching the guarantees of SoS.}
From a theoretical standpoint, it reduces the task of showing lower bounds against the complicated sum-of-squares SDP to that of understanding the spectrum of low-degree matrix polynomials over the two distributions.

\paragraph{Future directions.}
The connection in \cref{thm:spectralSDP} could potentially be tightened, leading to a fine-grained understanding of the power of sum-of-squares SDPs.
We will use a concrete example to expound on the questions suggested by \cref{thm:spectralSDP}, but the discussion is applicable more broadly too.

Consider the problem of certifying an upper bound on the size of maximum independent sets in sparse random graphs.
Formally, let $G$ be a sparse random graph drawn from $\mathbb{G}(n, k/n)$ by sampling each edge independently with probability $k/n$.
There exists a constant $\alpha_k \in (0,1)$ such that the size of the largest independent set in $G$ is $(\alpha_k \pm o(1)) \cdot n$ with high probability.
For every $\beta \in (0,1)$, the existence of a size $\beta \cdot n$-independent set can be formulated as the following polynomial system.
\[ \cP_{\beta}(G): \left\{ \{ x_i(1 - x_i) = 0\}_{i \in [n]} , \qquad \{ x_i x_j = 0 \}_{(i,j) \in E(G)} ,\qquad \sum_{i \in [n]} x_i \geq \beta \cdot n \mper \right\} \]
For each degree-$d \in \N$ define the degree-$d$ SoS SDP refutation threshold to be
\[ \alpha^{(d)}_k \defeq \text{ smallest } \beta \text{ such that }  \Pr_{G\sim \mathbb{G}([n],k/n)}\left[ \cP_\beta(G) \sststile{d}{x} \bot \right] = 1 - o_n(1)
\]

It is natural to ask if the degree-$d$ sum-of-squares SDP refutation threshold steadily improves with $k$.
\begin{question}
    Is $\{ \alpha^{(d)}_k\}_{k \in \N}$ a strictly decreasing sequence?
\end{question}

A natural structured distribution $\cD_{\beta}$ for the problem is the following:
For each subset $S \in \binom{[n]}{\beta \cdot n}$, define $\mu_S$ as $\mathbb{G}(n,k/n)$ conditioned on $S$ being an independent set.
For $D \in \N$ let $\gamma^{(D)}_k \in (0,1)$ be the largest value of $\beta$ for which distribution of eigenvalues of every degree-$D$ matrix polynomial in the structured distribution $\cD_\beta$ and null distribution $\randomdist$ converge to each other in distribution.
In other words, $\gamma^{(D)}_k$ is the precise threshold of independent set size $\beta$ below which the spectrum of degree-$D$ matrix polynomials fails to distinguish the structured and null distributions.
It is natural to conjecture that if the empirical distribution of eigenvalues looks alike then the sum-of-squares SDP cannot distinguish between the two.
Roughly speaking, the conjecture formalizes the notion that sum-of-squares SDPs are no more powerful than spectral algorithms.

\begin{question}\label[question]{q:spect-sos}
    Is there a universal constant $C > 0$ such that $\alpha^{(d)}_k \geq \gamma^{(C \cdot d)}_{k}$?
\end{question}

This question strengthens \cref{thm:spectralSDP} in that we ask for the degree $d$ of the SoS refutation to differ from the degree of the spectral algorithm by only a {\em constant} factor; in \cref{thm:spectralSDP}, the degrees differ by a factor that depends on the robustness of the polynomial system, which may grow with $n$.
On the other hand, this question differs from \cref{thm:spectralSDP} because if we ask for the spectra of matrices from the structured and null distributions to converge in distribution, we are working with a different class of spectral algorithms.
Depending on the notion of convergence, we may not be able to reason about the value of the maximum positive eigenvalue, or other non-smooth tests.
\cref{q:spect-sos} and its variants are an intriguing direction for future research.

%% file: content/conclusion.tex
\section{Concluding remarks}

We have now seen how the sum-of-squares algorithm may be used as a tool for solving estimation problems, via low-degree SoS proofs of identifiability of parameters from measurements.
This proofs-to-algorithms perspective has unified and simplified previous algorithmic results, as well as lead to stronger novel ones.

On the other hand, we have surveyed recent progress towards characterizing the limitations of SoS algorithms for estimation problems in the regime where the estimation problem is information-theoretically solvable (but perhaps not computationally tractable).
The {\em pseudocalibration} heuristic of \cite{DBLP:conf/focs/BarakHKKMP16} suggests exposing the limitations of SoS by comparing {\em structured} and {\em null} distributions over measurements; when SoS fails to distinguish these distributions, the SoS algorithm fails to solve the estimation problem.
But many questions remain: which properties of the structured distribution dictate whether low-degree SoS proofs of identifiability (and therefore algorithms) exist?
How does the degree of SoS proofs scale with the amount of information in the measurement?

Remarkably, all current evidence is consistent with the conjecture that low-degree sum-of-squares proofs are only as powerful as low-degree polynomial tests for a broad family of estimation problems (\cref{conj:low-deg}).
Affirming this conjecture will establish a beautiful theory of the power of semidefinite programs, and bring new insight to the study of information-computation gaps.
Refuting this conjecture may lead to exciting algorithmic discoveries, and a fine-grained understanding of the difficulty of estimation problems.

\section*{Acknowledgements}
We would like to thank Samuel B. Hopkins for pointing out a small error in a previous version of this survey.

%% file: content/algproofs.tex
\section{Continued proofs of identifiability}\label[appendix]{app:alg-proofs}

Here, we fill in some details from the proofs in \cref{sec:algorithms}.

\begin{claim*}[Restatement of \cref{claim:eq-cube}]
    When $\{a_i\}_{i\in[r]}$ are orthogonal and
    \begin{enumerate}
	\item $\cA \sststile{}{} \left\{\textstyle{\sum_{j \in [r]}} \iprod{a_j,b_i}^2 = 1\right\}_{i \in [r]}$, and
	\item $\cA \sststile{}{} \set{\iprod{a_{j_1},b_i}^2\iprod{a_{j_2},b_i}^2 = 0 }_{j_1 \neq j_2 \in [r]}$,
    \end{enumerate}
    then
    \[
	\cA \sststile{}{} \norm{b_i^{\otimes 3}-\sum_{j=1}^r \iprod{a_j,b_i}^3 a_j^{\otimes 3}}^2 =0.
	\]
\end{claim*}
\begin{proof}
    It follows from the orthogonality of the $a_i$ and from the first condition that $\cA \sststile{}{} \norm{b_i - \sum_{j=1}^r \iprod{a_j,b_i} a_j}^2=0$.
  To prove the claim, we now use that in turn, $\cA \sststile{}{} \norm{b_i^{\otimes 3} - \paren{\sum_{j=1}^r \iprod{a_j,b_i} a_j}^{\otimes 3}}^2=0$ and verify
  \begin{displaymath}
    \cA \sststile {}{}\quad
    \begin{aligned}[t]
      & \Norm{\Paren{\textstyle \sum_{j=1}^r \iprod{a_j,b_i} a_j}^{\otimes 3} - \textstyle \sum_{j=1}^r \iprod{a_j,b_i}^3 \cdot a_j^{\otimes 3}}^2\\
      & = \Norm{\sum_{\substack{j_1,j_2,j_3\in [r]\\\text{not all equal}}}\iprod{a_{j_1},b_i}\iprod{a_{j_2},b_i}\iprod{a_{j_3},b_i} \cdot a_{j_1}\otimes a_{j_2}\otimes a_{j_3}}^2\\
	& \le C(r) \sum_{\substack{j_1,j_2,j_3\in [r]\\\text{not all equal}}} \iprod{a_{j_1},b_i}^2\iprod{a_{j_2},b_i}^2\iprod{a_{j_3},b_i}^2\\
      & = 0\,.
    \end{aligned}
  \end{displaymath}
    Where we have used that $\sststile{}{x} \set{(x_1+\cdots+x_r)^2 \le C(r)\cdot (x_1^2 + \cdots + x_r^2)}$ for some function $C(r)$, and the second condition of the claim.
  We conclude that
  \begin{displaymath}
    \cA \sststile {}{} \normf{\sum_{i=1}^r b_i^{\otimes 3} - \sum_{i,j} \iprod{a_j,b_i}^3 a_j^{\otimes 3}}^2 = 0\,.
  \end{displaymath}
    as desired.
\end{proof}

We now give a full proof of the robust version of Jennrich's algorithm.
\restatetheorem{robust-jennrich}

\begin{proof}
    For a vector $v \in \R^n$, define the linear operator $\cM_{v} = \Id \otimes \Id \otimes v^\top$ from $(\R^{n})^{\otimes 3} \to (\R^n)^{\otimes 2}$.
  We apply the following version of Jennrich's algorithm to $\bT$:
  Choose a Gaussian vector $g \sim \cN(0,\Id)$ and apply $\cM_g$ to the $d^3 \times 1$ reshaping of $\bT$.
    Reshape the resulting vector $\cM_g \bT$ to an $n\times n$ matrix:
    \[
	(\cM_g \bT)_{\set{1}\set{2}} = \sum_{j \in [n]}g_j \cdot T_i,
    \]
    where $T_i$ is the $n \times n$ matrix resulting from the restriction of $\bT$ to coordinate $i$ in the third mode.
    Then, output the top eigenvector of $(\cM_g \bT)_{\set{1}\set{2}}$.

    We let $\bT = S + E$ where $S = \sum_i a_i^{\otimes 3}$.
  First, we claim that $\normf{\cM_{a_i} E}^2\le 4 \e$ for at least half of the indices $i\in [r]$.
  This claim follows by averaging from the following bound
  \begin{multline}
    \sum_{i=1}^r \normf{\cM_{a_i} E}^2 = \Iprod{E, \Paren{\textstyle \sum_{i=1}^r \transpose{\cM_{a_i}}\cM_{a_i}} E}\\
    \le \normf{E}^2\cdot \lmax\Paren{\textstyle \sum_{i=1}^r \transpose{\cM_{a_i}}\cM_{a_i}}
    =  \normf{E}^2\cdot \lmax(\dyad A)\le (1+\e)\e r\le 2 \e r\,.
  \end{multline}
  Second, we claim $\norm{(\cM_{a_i} S)_{\set{1}\set{2}}-\dyad{ a_i}} \le 2\e$ for all $i\in [r]$.
  Indeed, by our assumption that the orthogonality defect is bounded,
    \[
	(\cM_{a_i} S)_{\set{1}\set{2}}-\dyad{a_i}=\sum_{j\neq i}\iprod{a_i,a_j}\cdot \dyad {a_j}\preceq \e \cdot \dyad A \preceq (1+\e)\e \cdot \Id
	\]
    and, by the same reasoning, $(\cM_{a_i} S)_{\set{1}\set{2}}-\dyad{a_i}\succeq -(1+\e)\e\cdot \Id$.
  Taken together, these bounds imply $\norm{(\cM_{a_i} \bT)_{\set 1 \set 2} - \dyad {a_i}}\le 6\e$ for at least half of the indices $i\in [r]$.

  To finish the analysis, we consider an index $i\in [r]$ that satisfies  $\norm{(\cM_{a_i} \bT)_{\set{1}\set 2} - \dyad {a_i}}\le 6\e$.
    Decomposing $g = \sum_{i} \iprod{a_i,g}\cdot a_i + g'$, we write $\cM_g \bT$ as the sum $\cM_{g}\bT = \iprod{g,a_i}\cM_{a_i} \bT + \cM_{g'} \bT$.
  By a matrix Chernoff bound, the assumption $\max\set{\norm{\bT}_{\set{1,3}\set{2}},\norm{\bT}_{\set{1}\set{2,3}}}\le 10$ implies that with high probability $\norm{\cM_{g'} \bT}_{\set 1 \set 2}\le O(\sqrt{\log n})$.
  (See \cite{DBLP:conf/focs/MaSS16} for details.)
  Since $\iprod{g,a_i}$ is independent of $g'$, the event $\iprod{g,a_i}\ge 1/\e \cdot \norm{\cM_{g'} \bT}_{\set 1 \set 2}$ has inverse polynomial probability in $n$ (but exponentially small probability in $1/\e$).
  It is straightforward to verify that in this event $\tfrac 1 {\iprod{g,a_i}}\cdot (\cM_{g}\bT)_{\set 1 \set 2}$ is at most $7\e$ far from $\dyad {a_i}$ in spectral norm.
  For small enough $\e$, these events imply that the algorithm outputs a unit vector $u$ that satisfies the conclusion of the theorem.
\end{proof}

%% file: content/lbproofs.tex
\section{Approximate preservation of equalities under low-degree projection}
\label[appendix]{app:proof-of}
We now prove \cref{thm:pseudocalib-constraints}.
\restatetheorem{thm:pseudocalib-constraints}

We begin with a couple of observations.
Notice that $L_2(\jndist) = L_2(\nulld) \otimes L_2(\xdist)$ since the distribution $\jndist(x,y) = \nulld(y) \cdot \xdist(x)$.
By Gram-Schmidt orthogonalization, the vector space $L_2(\nulld)$ can be written as a direct sum of
\[L_2(\nulld) = \bigoplus_{i = 0}^{\infty} \cY_i \]
such that the following properties hold:
\begin{enumerate}
    \item If $f \in \cY_i$ and $g \in \cY_j$ with $i < j$, then $\iprod{f,g}_{\nulld} = 0$.
    \item For each $i\in \N$ and $f \in \cY_i$, $\deg_y(f) \leq i$.
\end{enumerate}
Similarly, one can decompose $L_2(\sigma) = \bigoplus_{i= 0}^{\infty} \cX_i$ with similar properties.
Abusing notation, we will use $\opY_i$ to also denote the projection operator on to the space $\cY_i$.
Let $\opY_{\leq D}$ denote the projection operator on to the space $\bigoplus_{i=0}^{D} \cY_i$, and let $\opY_{[a,b]}$ denote the projection on to $\bigoplus_{i = a}^{b} \cY_i$.
Let $\opX_{\leq d}$, $\cX_{[a,b]}$ be similar operators for $L_2(\xdist)$.
Both $\opY_{\leq D}$ and $\opX_{\leq d}$ operators have a natural action on tensor product space $L_2(\jndist)$.
In fact, the pseudocalibrated function can be defined as  $\bar \mu = \opY_{\leq D} \circ \opX_{\leq d} \circ \mu_*$.

We will require the following lemma, which relates the product of projections of polynomials to the projection of their product:
\begin{lemma} \label{lem:identity}
    Suppose $f,g \in L_2(\jndist)$ are polynomials and that $\deg_y(g) = D_y$.
    Then the following relationship between the the projection of the product and the product of projections holds:
    \begin{equation}
	\opY_{\leq D+D_y} \circ (f g)  = (\opY_{\leq D} \circ f ) g   + \opY_{\leq D+D_y} \circ \left((\opY_{[D,D+2D_y]} \circ f) g\right)\mper\label{eq:proj-prod}
    \end{equation}
\end{lemma}

\begin{proof}
    Using the decomposition of the projector $\Pi^{\cY} = \Pi_{< D}^{\cY} + \Pi_{[D,D+2D_y]}^{\cY} + \Pi_{> D + 2D_y}^{\cY}$ and that $f \in L_2(\jndist)$,
    we can express the left-hand-side term of \cref{eq:proj-prod} as
\begin{align*}
    \opY_{\leq D+D_y} \circ (f g) & = \opY_{\leq D+D_y} \circ \left( \left(\opY_{< D} \circ f + \opY_{[D,D+2D_y]} \circ f + \opY_{> D+2D_y} \circ f\right) g \right)\mcom\\
                & =  \opY_{\leq D+D_y} \circ \left((\opY_{\leq D} \circ f ) g\right) ~+~ \opY_{\leq D+D_y} \circ \left((\opY_{[D,D+2D_y]} \circ f) g\right) ~+~ (\opY_{> D+2D_y} \circ f)g\mper
\end{align*}
    In the first term, $\deg( (\opY_{\leq D} \circ f) g) \leq D + D_y$, so we may drop the leading projector to obtain the first right-hand-side of \cref{eq:proj-prod}.
For the third term, note that
\[ \opY_{\leq D+D_y} \circ \left( (\opY_{> D+2D_y} \circ f) g \right) = 0,\]
    for otherwise, we would have a polynomial $h \in \cY_{\leq D+D_y}$ such that
    \[
	\Iprod{h, \left(\opY_{>D+2D_y} \circ f \right)g}_{\nulld} = \Iprod{hg, \opY_{>D+2D_y} \circ f} \neq 0\mcom
	\]
	a contradiction since $\deg_y(hg) \leq D+2D_y$ while $\opY_{> D+2D_y} \circ f \in \cY_{ > D+2D_y}$.
    Putting these facts together, \cref{eq:proj-prod} follows immediately.
\end{proof}

    \begin{proof}[Proof of \cref{thm:pseudocalib-constraints}]
Set $f = \opX_{\leq d} \circ \mu_*$ and $g = p$ in the statement of \cref{lem:identity}.
Rearranging, we have
	\[
	(\opY_{\leq D} \circ \opX_{\leq d} \circ \mu_* ) p
	=   \opY_{\leq D+D_y } \circ \left((\opX_{\leq d} \circ \mu_*) p\right) -   \opY_{\leq D+D_y} \circ \left( (\opY_{[D,D+2D_y]} \circ \opX_{\leq d} \circ \mu_*) p\right)    \mcom
	\]
	and applying the $\E_x$ operator on both sides, we get
	\[
	\E_{x} \left[(\opY_{\leq D} \circ \opX_{\leq d} \circ \mu_* ) p\right]
	=   \opY_{\leq D+D_y } \circ \left(\E_x \left[(\opX_{\leq d} \circ \mu_*) p\right]\right) -   \opY_{\leq D+D_y} \circ \left(\E_{x} \left[(\opY_{[D,D+2D_y]} \circ \opX_{\leq d} \circ \mu_*) p\right]\right)    \mper
	\]
By definition of the pseudocalibrated function $\bar \mu_y$, the left hand side is equal to $ \E_{x} p(x,y) \bar \mu_y(x)$.
Since $\deg_x(p) \leq d$, we have $\E_x (\opX_{\leq d} \circ \mu_*) p = \E_x \mu_* p$.  Further $p(x,y) = 0$ for each $(x,y) \in \supp(\mu_*)$ implies that that $\E_x \mu_*(x,y) p(x,y) = 0$ for all $y$.
Thus the first term in the right-hand side, given by $\opY_{\leq D+D_y} \circ \left( \E_x \left[(\opX_{\leq d} \circ \mu_*) p\right]\right)$, is $0$.
Therefore we have the following inequality for each $y$,
\[ \E_{x} p(x,y) \bar \mu_y(x)
	\geq - \opY_{\leq D+D_y} \circ \left(\E_{x} \left[(\opY_{[D,D+2D_y]} \circ \opX_{\leq d} \circ \mu_*) p\right]\right)\mper \]
Now, we apply Chebyshev's inequality to the right-hand side of the above,
\begin{align*}
    \Pr_{y \sim \nulld}[ \abs{\E_{x} p(x,y) \bar \mu_y(x)}  \geq \epsilon ]
    & \leq \frac{1}{\epsilon^2} \E_{y \sim \nulld} \left(\opY_{\leq D+D_y} \circ \left(\E_{x} (\opY_{[D,D+2D_y]} \circ \opX_{\leq d} \circ \mu_*) p\right)\right)^2 \mcom \\
    &= \frac{1}{\epsilon^2}  \left\lVert\opY_{\leq D+D_y} \circ (\E_{x} (\opY_{[D,D+2D_y]} \circ (\opX_{\leq d} \circ \mu_*)) p)\right\rVert_{2,\jndist}^2\mcom
    \intertext{And since norms decrease under projection,}
    & \leq \frac{1}{\epsilon^2}  \left\lVert \E_{x} (\opY_{[D,D+2D_y]} \circ \opX_{\leq d} \circ \mu_*) p\right\rVert_{2,\jndist}^2\mcom
    \intertext{Now, since we have assumed that $\max_{(x,y) \in \jndist} |p(x,y)| \le B$,}
& \leq \frac{B^2}{\epsilon^2} \E_{y \sim \nulld} \left ( \E_{x} \left|(\opY_{[D,D+2D_y]} \circ \opX_{\leq d} \circ \mu_*)\right| \right )^2\mcom \\
& \leq   \frac{B^2}{\epsilon^2} \E_{y \sim \nulld}  \E_{x} \left|(\opY_{[D,D+2D_y]} \circ \opX_{\leq d} \circ \mu_*)\right|^2\mcom \\
&=   \frac{B^2}{\epsilon^2} \Norm{\opY_{[D,D+2D_y]} \circ \opX_{\leq d} \circ \mu_*}_{2,\jndist}^2 \mper
\end{align*}
This concludes the proof.
\end{proof}

%% file: main.bbl
\newcommand{\etalchar}[1]{$^{#1}$}
\providecommand{\bysame}{\leavevmode\hbox to3em{\hrulefill}\thinspace}
\providecommand{\MR}{\relax\ifhmode\unskip\space\fi MR }
\providecommand{\MRhref}[2]{%
  \href{http://www.ams.org/mathscinet-getitem?mr=#1}{#2}
}
\providecommand{\href}[2]{#2}
\begin{thebibliography}{KMOW17}

\bibitem[ABA{\u{C}}]{doi:10.1002/cpa.21422}
Antonio Auffinger, G\'{e}rard Ben~Arous, and Ji\u{r}\'{i} {\u{C}}eran\'{y},
  \emph{Random matrices and complexity of spin glasses}, Communications on Pure
  and Applied Mathematics \textbf{66}, no.~2, 165--201.

\bibitem[AFH{\etalchar{+}}12]{DBLP:conf/nips/AnandkumarFHKL12}
Anima Anandkumar, Dean~P. Foster, Daniel~J. Hsu, Sham Kakade, and Yi{-}Kai Liu,
  \emph{A spectral algorithm for latent dirichlet allocation}, {NIPS}, 2012,
  pp.~926--934.

\bibitem[AGMR17]{DBLP:conf/stoc/AroraGMR17}
Sanjeev Arora, Rong Ge, Tengyu Ma, and Andrej Risteski, \emph{Provable learning
  of noisy-or networks}, Proceedings of the 49th Annual {ACM} {SIGACT}
  Symposium on Theory of Computing, {STOC} 2017, Montreal, QC, Canada, June
  19-23, 2017 (Hamed Hatami, Pierre McKenzie, and Valerie King, eds.), {ACM},
  2017, pp.~1057--1066.

\bibitem[AK01]{DBLP:conf/stoc/SanjeevK01}
Sanjeev Arora and Ravi Kannan, \emph{Learning mixtures of arbitrary gaussians},
  {STOC}, {ACM}, 2001, pp.~247--257.

\bibitem[AKS98a]{DBLP:journals/rsa/AlonKS98}
Noga Alon, Michael Krivelevich, and Benny Sudakov, \emph{Finding a large hidden
  clique in a random graph}, Random Struct. Algorithms \textbf{13} (1998),
  no.~3-4, 457--466.

\bibitem[AKS98b]{MR1662795-Alon98}
Noga Alon, Michael Krivelevich, and Benny Sudakov, \emph{Finding a large hidden
  clique in a random graph}, Proceedings of the {E}ighth {I}nternational
  {C}onference ``{R}andom {S}tructures and {A}lgorithms'' ({P}oznan, 1997),
  vol.~13, 1998, pp.~457--466. \MR{1662795}

\bibitem[AM05]{DBLP:conf/colt/AchlioptasM05}
Dimitris Achlioptas and Frank McSherry, \emph{On spectral learning of mixtures
  of distributions}, {COLT}, Lecture Notes in Computer Science, vol. 3559,
  Springer, 2005, pp.~458--469.

\bibitem[AOW15]{DBLP:conf/focs/AllenOW15}
Sarah~R. Allen, Ryan O'Donnell, and David Witmer, \emph{How to refute a random
  {CSP}}, {FOCS}, {IEEE} Computer Society, 2015, pp.~689--708.

\bibitem[BB02]{DBLP:journals/eccc/ECCC-TR02-003}
Eli Ben{-}Sasson and Yonatan Bilu, \emph{A gap in average proof complexity},
  Electronic Colloquium on Computational Complexity {(ECCC)} (2002), no.~003.

\bibitem[BBH{\etalchar{+}}12]{DBLP:conf/stoc/BarakBHKSZ12}
Boaz Barak, Fernando G. S.~L. Brand{\~{a}}o, Aram~Wettroth Harrow, Jonathan~A.
  Kelner, David Steurer, and Yuan Zhou, \emph{Hypercontractivity,
  sum-of-squares proofs, and their applications}, {STOC}, {ACM}, 2012,
  pp.~307--326.

\bibitem[BCK15]{MR3388187-Barak15}
Boaz Barak, Siu~On Chan, and Pravesh~K. Kothari, \emph{Sum of squares lower
  bounds from pairwise independence [extended abstract]},
  S{TOC}'15---{P}roceedings of the 2015 {ACM} {S}ymposium on {T}heory of
  {C}omputing, ACM, New York, 2015, pp.~97--106. \MR{3388187}

\bibitem[BCMV14]{DBLP:conf/stoc/BhaskaraCMV14}
Aditya Bhaskara, Moses Charikar, Ankur Moitra, and Aravindan Vijayaraghavan,
  \emph{Smoothed analysis of tensor decompositions}, {STOC}, {ACM}, 2014,
  pp.~594--603.

\bibitem[BCV{\etalchar{+}}12]{DBLP:conf/soda/BhaskaraCVGZ12}
Aditya Bhaskara, Moses Charikar, Aravindan Vijayaraghavan, Venkatesan
  Guruswami, and Yuan Zhou, \emph{Polynomial integrality gaps for strong {SDP}
  relaxations of densest \emph{k}-subgraph}, {SODA}, {SIAM}, 2012,
  pp.~388--405.

\bibitem[BGL17]{DBLP:conf/approx/BhattiproluGL17}
Vijay Bhattiprolu, Venkatesan Guruswami, and Euiwoong Lee, \emph{Sum-of-squares
  certificates for maxima of random tensors on the sphere}, {APPROX-RANDOM},
  LIPIcs, vol.~81, Schloss Dagstuhl - Leibniz-Zentrum fuer Informatik, 2017,
  pp.~31:1--31:20.

\bibitem[BHK{\etalchar{+}}16]{DBLP:conf/focs/BarakHKKMP16}
Boaz Barak, Samuel~B. Hopkins, Jonathan~A. Kelner, Pravesh Kothari, Ankur
  Moitra, and Aaron Potechin, \emph{A nearly tight sum-of-squares lower bound
  for the planted clique problem}, {FOCS}, {IEEE} Computer Society, 2016,
  pp.~428--437.

\bibitem[BKS13]{DBLP:conf/innovations/BarakKS13}
Boaz Barak, Guy Kindler, and David Steurer, \emph{On the optimality of
  semidefinite relaxations for average-case and generalized constraint
  satisfaction}, {ITCS}, {ACM}, 2013, pp.~197--214.

\bibitem[BKS14]{DBLP:conf/stoc/BarakKS14}
Boaz Barak, Jonathan~A. Kelner, and David Steurer, \emph{Rounding
  sum-of-squares relaxations}, {STOC}, {ACM}, 2014, pp.~31--40.

\bibitem[BKS15]{DBLP:conf/stoc/BarakKS15}
\bysame, \emph{Dictionary learning and tensor decomposition via the
  sum-of-squares method}, {STOC}, {ACM}, 2015, pp.~143--151.

\bibitem[BM16]{DBLP:conf/colt/BarakM16}
Boaz Barak and Ankur Moitra, \emph{Noisy tensor completion via the
  sum-of-squares hierarchy}, {COLT}, {JMLR} Workshop and Conference
  Proceedings, vol.~49, JMLR.org, 2016, pp.~417--445.

\bibitem[BS10]{DBLP:conf/colt/BelkinS10}
Mikhail Belkin and Kaushik Sinha, \emph{Toward learning gaussian mixtures with
  arbitrary separation}, {COLT}, Omnipress, 2010, pp.~407--419.

\bibitem[BS14]{DBLP:journals/eccc/BarakS14}
Boaz Barak and David Steurer, \emph{Sum-of-squares proofs and the quest toward
  optimal algorithms}, Electronic Colloquium on Computational Complexity
  {(ECCC)} \textbf{21} (2014), 59.

\bibitem[Che15]{DBLP:journals/tit/Chen15}
Yudong Chen, \emph{Incoherence-optimal matrix completion}, {IEEE} Trans.
  Information Theory \textbf{61} (2015), no.~5, 2909--2923.

\bibitem[CLP02]{crisanti20023}
Andrea Crisanti, Luca Leuzzi, and Giorgio Parisi, \emph{The 3-{SAT} problem
  with large number of clauses in the $\infty$-replica symmetry breaking
  scheme}, Journal of Physics A: Mathematical and General \textbf{35} (2002),
  no.~3, 481.

\bibitem[CO12]{DBLP:journals/siammax/ChiantiniO12}
Luca Chiantini and Giorgio Ottaviani, \emph{On generic identifiability of
  3-tensors of small rank}, {SIAM} J. Matrix Analysis Applications \textbf{33}
  (2012), no.~3, 1018--1037.

\bibitem[CR09]{DBLP:journals/focm/CandesR09}
Emmanuel~J. Cand{\`{e}}s and Benjamin Recht, \emph{Exact matrix completion via
  convex optimization}, Foundations of Computational Mathematics \textbf{9}
  (2009), no.~6, 717--772.

\bibitem[Das99]{DBLP:conf/focs/Dasgupta99}
Sanjoy Dasgupta, \emph{Learning mixtures of gaussians}, {FOCS}, {IEEE} Computer
  Society, 1999, pp.~634--644.

\bibitem[DKS18]{DiakonikolasKS18}
Ilias Diakonikolas, Daniel~M. Kane, and Alistair Stewart, \emph{List-decodable
  robust mean estimation and learning mixtures of spherical gaussians mixture
  models, robustness, and sum of squares proofs}, {STOC}, {ACM}, 2018, p.~(to
  appear).

\bibitem[DLS14]{DBLP:conf/stoc/DanielyLS14}
Amit Daniely, Nati Linial, and Shai Shalev{-}Shwartz, \emph{From average case
  complexity to improper learning complexity}, {STOC}, {ACM}, 2014,
  pp.~441--448.

\bibitem[DM15]{DBLP:conf/colt/DeshpandeM15}
Yash Deshpande and Andrea Montanari, \emph{Improved sum-of-squares lower bounds
  for hidden clique and hidden submatrix problems}, {COLT}, {JMLR} Workshop and
  Conference Proceedings, vol.~40, JMLR.org, 2015, pp.~523--562.

\bibitem[Fei02]{MR2121179-Feige02}
Uriel Feige, \emph{Relations between average case complexity and approximation
  complexity}, Proceedings of the {T}hirty-{F}ourth {A}nnual {ACM} {S}ymposium
  on {T}heory of {C}omputing, ACM, New York, 2002, pp.~534--543. \MR{2121179}

\bibitem[FGR{\etalchar{+}}17]{DBLP:journals/jacm/FeldmanGRVX17}
Vitaly Feldman, Elena Grigorescu, Lev Reyzin, Santosh~Srinivas Vempala, and
  Ying Xiao, \emph{Statistical algorithms and a lower bound for detecting
  planted cliques}, J. {ACM} \textbf{64} (2017), no.~2, 8:1--8:37.

\bibitem[FK00]{MR1742351-Feige00}
Uriel Feige and Robert Krauthgamer, \emph{Finding and certifying a large hidden
  clique in a semirandom graph}, Random Structures Algorithms \textbf{16}
  (2000), no.~2, 195--208. \MR{1742351}

\bibitem[FK03]{MR1969394-Feige03}
\bysame, \emph{The probable value of the {L}ov\'asz-{S}chrijver relaxations for
  maximum independent set}, SIAM J. Comput. \textbf{32} (2003), no.~2,
  345--370. \MR{1969394}

\bibitem[FS02]{MR1900615-Feige02}
Uriel Feige and Gideon Schechtman, \emph{On the optimality of the random
  hyperplane rounding technique for {MAX} {CUT}}, Random Structures Algorithms
  \textbf{20} (2002), no.~3, 403--440, Probabilistic methods in combinatorial
  optimization. \MR{1900615}

\bibitem[Gha10]{gharibian2010strong}
Sevag Gharibian, \emph{Strong np-hardness of the quantum separability problem},
  Quantum Information \& Computation \textbf{10} (2010), no.~3, 343--360.

\bibitem[GM75]{grimmett_mcdiarmid_1975}
G.~R. Grimmett and C.~J.~H. McDiarmid, \emph{On colouring random graphs},
  Mathematical Proceedings of the Cambridge Philosophical Society \textbf{77}
  (1975), no.~2, 313–324.

\bibitem[GM15]{DBLP:conf/approx/GeM15}
Rong Ge and Tengyu Ma, \emph{Decomposing overcomplete 3rd order tensors using
  sum-of-squares algorithms}, {APPROX-RANDOM}, LIPIcs, vol.~40, Schloss
  Dagstuhl - Leibniz-Zentrum fuer Informatik, 2015, pp.~829--849.

\bibitem[Gri01a]{DBLP:journals/cc/Grigoriev01}
Dima Grigoriev, \emph{Complexity of positivstellensatz proofs for the
  knapsack}, Computational Complexity \textbf{10} (2001), no.~2, 139--154.

\bibitem[Gri01b]{DBLP:journals/tcs/Grigoriev01}
\bysame, \emph{Linear lower bound on degrees of positivstellensatz calculus
  proofs for the parity}, Theor. Comput. Sci. \textbf{259} (2001), no.~1-2,
  613--622.

\bibitem[Gro11]{DBLP:journals/tit/Gross11}
David Gross, \emph{Recovering low-rank matrices from few coefficients in any
  basis}, {IEEE} Trans. Information Theory \textbf{57} (2011), no.~3,
  1548--1566.

\bibitem[GS]{godsil2010association}
Chris Godsil and Sung~Y Song, \emph{Association schemes}, Handbook of
  Combinatorial Designs,, 325--330.

\bibitem[Gur03]{gurvits2003classical}
Leonid Gurvits, \emph{Classical deterministic complexity of edmonds' problem
  and quantum entanglement}, Proceedings of the thirty-fifth annual ACM
  symposium on Theory of computing, ACM, 2003, pp.~10--19.

\bibitem[GW95]{MR1412228-Goemans95}
Michel~X. Goemans and David~P. Williamson, \emph{Improved approximation
  algorithms for maximum cut and satisfiability problems using semidefinite
  programming}, J. Assoc. Comput. Mach. \textbf{42} (1995), no.~6, 1115--1145.
  \MR{1412228}

\bibitem[Har70]{harshman1970foundations}
Richard~A Harshman, \emph{Foundations of the parafac procedure: Models and
  conditions for an" explanatory" multi-modal factor analysis}.

\bibitem[HK13]{DBLP:conf/innovations/HsuK13}
Daniel~J. Hsu and Sham~M. Kakade, \emph{Learning mixtures of spherical
  gaussians: moment methods and spectral decompositions}, {ITCS}, {ACM}, 2013,
  pp.~11--20.

\bibitem[HKP{\etalchar{+}}16]{DBLP:conf/soda/HopkinsKPRS16}
Samuel~B. Hopkins, Pravesh Kothari, Aaron~Henry Potechin, Prasad Raghavendra,
  and Tselil Schramm, \emph{On the integrality gap of degree-4 sum of squares
  for planted clique}, {SODA}, {SIAM}, 2016, pp.~1079--1095.

\bibitem[HKP{\etalchar{+}}17]{DBLP:conf/focs/HopkinsKPRSS17}
Samuel~B. Hopkins, Pravesh~K. Kothari, Aaron Potechin, Prasad Raghavendra,
  Tselil Schramm, and David Steurer, \emph{The power of sum-of-squares for
  detecting hidden structures}, {FOCS}, {IEEE} Computer Society, 2017,
  pp.~720--731.

\bibitem[HL13]{MR3144915-Hillar13}
Christopher~J. Hillar and Lek-Heng Lim, \emph{Most tensor problems are
  {NP}-hard}, J. ACM \textbf{60} (2013), no.~6, Art. 45, 39. \MR{3144915}

\bibitem[HL18]{HopkinsL18}
Sam~B. Hopkins and Jerry Li, \emph{Mixture models, robustness, and sum of
  squares proofs}, {STOC}, {ACM}, 2018, p.~(to appear).

\bibitem[HSS15]{DBLP:conf/colt/HopkinsSS15}
Samuel~B. Hopkins, Jonathan Shi, and David Steurer, \emph{Tensor principal
  component analysis via sum-of-square proofs}, {COLT}, {JMLR} Workshop and
  Conference Proceedings, vol.~40, JMLR.org, 2015, pp.~956--1006.

\bibitem[HSSS16]{DBLP:conf/stoc/HopkinsSSS16}
Samuel~B. Hopkins, Tselil Schramm, Jonathan Shi, and David Steurer, \emph{Fast
  spectral algorithms from sum-of-squares proofs: tensor decomposition and
  planted sparse vectors}, {STOC}, {ACM}, 2016, pp.~178--191.

\bibitem[Jer92]{DBLP:journals/rsa/Jerrum92}
Mark Jerrum, \emph{Large cliques elude the metropolis process}, Random Struct.
  Algorithms \textbf{3} (1992), no.~4, 347--360.

\bibitem[Kar76]{MR0445898-Karp76}
Richard~M. Karp, \emph{The probabilistic analysis of some combinatorial search
  algorithms}, 1--19. \MR{0445898}

\bibitem[KMOW17]{DBLP:conf/stoc/KothariMOW17}
Pravesh~K. Kothari, Ryuhei Mori, Ryan O'Donnell, and David Witmer, \emph{Sum of
  squares lower bounds for refuting any {CSP}}, {STOC}, {ACM}, 2017,
  pp.~132--145.

\bibitem[KMV10]{DBLP:conf/stoc/KalaiMV10}
Adam~Tauman Kalai, Ankur Moitra, and Gregory Valiant, \emph{Efficiently
  learning mixtures of two gaussians}, {STOC}, {ACM}, 2010, pp.~553--562.

\bibitem[Kri64]{krivine1964anneaux}
Jean-Louis Krivine, \emph{Anneaux pr{\'e}ordonn{\'e}s}, Journal d’analyse
  math{\'e}matique \textbf{12} (1964), no.~1, 307--326.

\bibitem[KSS18]{KothariSS18}
Pravesh~K. Kothari, Jacob Steinhardt, and David Steurer, \emph{Robust moment
  estimation and improved clustering via sum-of-squares}, {STOC}, {ACM}, 2018,
  p.~(to appear).

\bibitem[KV15]{MR3323774-Khot15}
Subhash~A. Khot and Nisheeth~K. Vishnoi, \emph{The unique games conjecture,
  integrability gap for cut problems and embeddability of negative-type metrics
  into {$\ell_1$}}, J. ACM \textbf{62} (2015), no.~1, Art. 8, 39. \MR{3323774}

\bibitem[Las01]{MR1814045-Lasserre00}
Jean~B. Lasserre, \emph{Global optimization with polynomials and the problem of
  moments}, SIAM J. Optim. \textbf{11} (2000/01), no.~3, 796--817. \MR{1814045}

\bibitem[LCC07]{DBLP:journals/tsp/LathauwerCC07}
Lieven~De Lathauwer, Jos{\'{e}}phine Castaing, and Jean{-}Fran{\c{c}}ois
  Cardoso, \emph{Fourth-order cumulant-based blind identification of
  underdetermined mixtures}, {IEEE} Trans. Signal Processing \textbf{55}
  (2007), no.~6-2, 2965--2973.

\bibitem[LRA93]{MR1238921-Leurgans93}
S.~E. Leurgans, R.~T. Ross, and R.~B. Abel, \emph{A decomposition for three-way
  arrays}, SIAM J. Matrix Anal. Appl. \textbf{14} (1993), no.~4, 1064--1083.
  \MR{1238921}

\bibitem[MP16]{DBLP:conf/approx/MedarametlaP16}
Dhruv Medarametla and Aaron Potechin, \emph{Bounds on the norms of uniform low
  degree graph matrices}, Approximation, Randomization, and Combinatorial
  Optimization. Algorithms and Techniques, {APPROX/RANDOM} 2016, September 7-9,
  2016, Paris, France (Klaus Jansen, Claire Mathieu, Jos{\'{e}} D.~P. Rolim,
  and Chris Umans, eds.), LIPIcs, vol.~60, Schloss Dagstuhl - Leibniz-Zentrum
  fuer Informatik, 2016, pp.~40:1--40:26.

\bibitem[MPW15]{DBLP:conf/stoc/MekaPW15}
Raghu Meka, Aaron Potechin, and Avi Wigderson, \emph{Sum-of-squares lower
  bounds for planted clique}, {STOC}, {ACM}, 2015, pp.~87--96.

\bibitem[MR05]{DBLP:conf/stoc/MosselR05}
Elchanan Mossel and S{\'{e}}bastien Roch, \emph{Learning nonsingular
  phylogenies and hidden markov models}, {STOC}, {ACM}, 2005, pp.~366--375.

\bibitem[MR14]{MontanariR14}
Andrea Montanari and Emile Richard, \emph{A statistical model for tensor pca},
  Proceedings of the 27th International Conference on Neural Information
  Processing Systems - Volume 2 (Cambridge, MA, USA), NIPS'14, MIT Press, 2014,
  pp.~2897--2905.

\bibitem[MS16]{MontanariS16}
Andrea Montanari and Nike Sun, \emph{Spectral algorithms for tensor
  completion}.

\bibitem[MSS16]{DBLP:conf/focs/MaSS16}
Tengyu Ma, Jonathan Shi, and David Steurer, \emph{Polynomial-time tensor
  decompositions with sum-of-squares}, {FOCS}, {IEEE} Computer Society, 2016,
  pp.~438--446.

\bibitem[MV10]{DBLP:conf/focs/MoitraV10}
Ankur Moitra and Gregory Valiant, \emph{Settling the polynomial learnability of
  mixtures of gaussians}, {FOCS}, {IEEE} Computer Society, 2010, pp.~93--102.

\bibitem[Par00]{parrilo2000structured}
Pablo~A Parrilo, \emph{Structured semidefinite programs and semialgebraic
  geometry methods in robustness and optimization}, Ph.D. thesis, California
  Institute of Technology, 2000.

\bibitem[PS17]{DBLP:conf/colt/PotechinS17}
Aaron Potechin and David Steurer, \emph{Exact tensor completion with
  sum-of-squares}, {COLT}, Proceedings of Machine Learning Research, vol.~65,
  {PMLR}, 2017, pp.~1619--1673.

\bibitem[Rec11]{DBLP:journals/jmlr/Recht11}
Benjamin Recht, \emph{A simpler approach to matrix completion}, Journal of
  Machine Learning Research \textbf{12} (2011), 3413--3430.

\bibitem[Rez00]{MR1747589-Reznick00}
Bruce Reznick, \emph{Some concrete aspects of {H}ilbert's 17th {P}roblem}, Real
  algebraic geometry and ordered structures ({B}aton {R}ouge, {LA}, 1996),
  Contemp. Math., vol. 253, Amer. Math. Soc., Providence, RI, 2000,
  pp.~251--272. \MR{1747589}

\bibitem[Rot13]{rothvoss2013lasserre}
Thomas Rothvo{\ss}, \emph{The lasserre hierarchy in approximation algorithms}.

\bibitem[RRS17]{DBLP:conf/stoc/RaghavendraRS17}
Prasad Raghavendra, Satish Rao, and Tselil Schramm, \emph{Strongly refuting
  random csps below the spectral threshold}, {STOC}, {ACM}, 2017, pp.~121--131.

\bibitem[Sch08]{DBLP:conf/focs/Schoenebeck08}
Grant Schoenebeck, \emph{Linear level lasserre lower bounds for certain
  k-csps}, {FOCS}, {IEEE} Computer Society, 2008, pp.~593--602.

\bibitem[SS17]{DBLP:conf/colt/SchrammS17}
Tselil Schramm and David Steurer, \emph{Fast and robust tensor decomposition
  with applications to dictionary learning}, {COLT}, Proceedings of Machine
  Learning Research, vol.~65, {PMLR}, 2017, pp.~1760--1793.

\bibitem[Ste74]{stengle1974nullstellensatz}
Gilbert Stengle, \emph{A nullstellensatz and a positivstellensatz in
  semialgebraic geometry}, Mathematische Annalen \textbf{207} (1974), no.~2,
  87--97.

\bibitem[Tao12]{tao2012topics}
Terence Tao, \emph{Topics in random matrix theory}, vol. 132, American
  Mathematical Soc., 2012.

\bibitem[Tre12]{MR2869009-Trevisan12}
Luca Trevisan, \emph{On {K}hot's unique games conjecture}, Bull. Amer. Math.
  Soc. (N.S.) \textbf{49} (2012), no.~1, 91--111. \MR{2869009}

\bibitem[Tul09]{DBLP:conf/stoc/Tulsiani09}
Madhur Tulsiani, \emph{{CSP} gaps and reductions in the lasserre hierarchy},
  {STOC}, {ACM}, 2009, pp.~303--312.

\bibitem[VW04]{DBLP:journals/jcss/VempalaW04}
Santosh Vempala and Grant Wang, \emph{A spectral algorithm for learning mixture
  models}, J. Comput. Syst. Sci. \textbf{68} (2004), no.~4, 841--860.

\end{thebibliography}
